\documentclass[journal]{IEEEtran}

\IEEEpubid{0000--0000/00\$00.00~\copyright~2015 IEEE}

\usepackage{amssymb, amsmath, amsthm, bm,amsfonts}
\usepackage{diagbox}
\usepackage{subfigure}
\newtheorem{lemma}{Lemma}
\newtheorem{theorem}{Theorem}
\newtheorem{definition}{Definition}
\newtheorem{Judgment criteria}{\bf Judgment criteria}
\newtheorem*{Contraction algorithm for the adjacency matrix}{\bf Contraction algorithm for the adjacency matrix}
\usepackage{multirow, bigstrut}
\usepackage{tabularx}
\usepackage{booktabs}
\usepackage{url}
\usepackage{array}
 \usepackage{color}
 \usepackage{hhline}
 \usepackage{tikz}
\usepackage{paralist}
\newcounter{example}[section]
\newenvironment{example}[1][]{\refstepcounter{example}\par\medskip
   \noindent \textbf{Example~\theexample. #1} \rmfamily}{\medskip}
\newcolumntype{x}[1]{>{\centering\arraybackslash\hspace{0pt}}p{#1}}
\newcolumntype{C}[1]{>{\centering}p{#1}}

\usepackage[normalem]{ulem}

\usepackage{graphicx}
\graphicspath{{figures-pdf/}}

\newlength\figwidth
\newlength\sfigwidth
\begin{document}
\title{Constructing Higher-Dimensional Digital Chaotic Systems via Loop-State Contraction Algorithm}
	
\author{Qianxue Wang, Simin Yu, Christophe Guyeux, and Wei Wang
\thanks{This work was supported by the National Natural Science
Foundation of China under Grants 61671161 and 61801127.}
\thanks{Q. Wang, S. Yu and W. Wang are with the College of Automation, Guangdong University
of Technology, Guangzhou 510006, Guangdong, China. (e-mail: wangqianxue@gdut.edu.cn, siminyu@163.com).}
\thanks{C. Guyeux is with the Femto-ST Institute, UMR 6174 CNRS, 
University of Bourgogne-Franche-Comt\'{e}, Besan\c{c}on 25000, France. (e-mail: christophe.guyeux@univ-fcomte.fr).}
\thanks{Copyright (c) 2015 IEEE. Personal use of this material is permitted. However, 
permission to use this material for any other purposes must be obtained from the
IEEE by sending an email to pubs-permissions@ieee.org.}
}
	
\markboth{IEEE Transactions on Circuits and Systems I}%
{Wang \MakeLowercase{\textit{et al.}}:}
		
\maketitle

\begin{abstract}
    In recent years, the generation of rigorously provable chaos in finite precision digital domain has made a lot of progress in theory and practice, this article is a part of it. It aims to improve and expand the theoretical and application framework of higher-dimensional digital chaotic system (HDDCS). In this study, topological mixing for HDDCS  is strictly proved theoretically at first. Topological mixing implies Devaney's definition of chaos in compact space, but not vise versa. Therefore, the proof of topological mixing improves and expands the theoretical framework of HDDCS. Then, a general design method for constructing HDDCS via loop-state contraction algorithm is given. The construction of the iterative function uncontrolled by random sequences (hereafter called iterative function) is the starting point of this research. On this basis, this article put forward a general design method to solve the problem of HDDCS construction, and a number of examples illustrate the effectiveness and feasibility of this method. The adjacency matrix corresponding to the designed HDDCS is used to construct the chaotic echo state network (ESN) for the Mackey-Glass time series prediction. The prediction accuracy is improved with the increase of the size of the reservoir and the dimension of HDDCS, indicating that higher-dimensional systems have better prediction performance than lower-dimensional systems.


\end{abstract}
	
\IEEEpeerreviewmaketitle

\section{Introduction}
\IEEEPARstart{T}{he} ability to generate a succession of chaotic systems on digital devices can be very useful, both for the simulation of physical phenomena and for the reinforcement of cryptographic technology on such machines~\cite{Li:ITCSI2019,Li:Cryptanalysis:IA2018,Hua:Design:IS2017,Ye:An:N2017,Lambic:A:ND2017}. At the mathematical level, the formalization of discrete dynamic systems with chaotic behaviour is theoretically mature~\cite{Xie:On the cryptanalysis:SP2017,MAY:Logistic:Nature1976,Chen:Design and FPGA-based:ITCSI2017,Zhou:Hidden:IJBC2018}, but its practical realization on any digital devices frequently suffers from the finite precision of the latter~\cite{LiShujun:Rules:IJBC2006,Hua:Sine:ITIE2018,Liu:An:ITCSI2016}. 
Without external control, for a chaotic autonomous system in finite precision digital domain, it cannot strictly satisfy the mathematical definition of chaos~\cite{Yu:Design Principles:2018}. Therefore, the fundamental way is to add some external control method to solve the modeling problem of digital chaos~\cite{Lai:Dynamic:CSF2018,Luo:A:ND2018,Liu:Counteracting:IJBC2017,Luo:encryption:IA2019,Wang:On fuzzy:ITC2014}. For example, in 2010, a new theory of constructing chaos in the digital domain was proposed in~\cite{Guyeux:Hash:JACT10}, that is, the chaos generation strategy controlled by random sequences. The main feature of this strategy is that some bits are randomly updated by the iterative function at each iteration, while the rest remains unchanged. This is very different from the existing discrete-time chaotic system in infinite precision real number domain, that all bits will participate in all the update operations by the iterative function. Digital chaotic system constructed with the chaos generation strategy controlled by random sequences is proven to satisfy the Devaney's definition of chaos in ~\cite{Guyeux:Hash:JACT10,Bahi:XORshift:JNCA04,Bahi:PRNS:IJAS11,SMYu:integer:IJBC2014,Wang:DCS:book18,Devaney:Chaos:2003}. Since this system is directly generated in finite precision digital domain, there is no finite precision effect problem, and the binary sequences can be obtained without transformation and combination. Therefore, it fundamentally solves the dynamic degradation problem of digital chaotic systems caused by the finite precision.

\IEEEpubidadjcol

Some works were firstly done to disclose how to build digital chaotic systems in finite precision domain, such as Chaotic Iterations (CI) system~\cite{Guyeux:Hash:JACT10}. Then, a mark sequence was applied to avoid wasteful duplication of values in~\cite{Bahi:XORshift:JNCA04}, leading to an obvious speed improvement. In~\cite{Bahi:PRNS:IJAS11}, chaotic combination of two input entropic streams has solved flaws exhibited in the system designed in~\cite{Bahi:XORshift:JNCA04}.
The chaos generation strategy was implemented in~\cite{SMYu:integer:IJBC2014}, through a sample-hold circuit and a decoder circuit so as to convert the uniform noise signal into a random sequence. The second chaos generation strategy named 
Chaotic Bitwise Dynamical System (CBDS) was proposed in~\cite{Wang:DCS:book18}. All the above related works deal with 1-D chaotic maps.

In a recent study, the concept of HDDCS based on the above chaos generation strategy controlled by random sequences has been proposed in~\cite{Wang:hddcs:IEEECSI2016}. It is simple and does not require floating-point operations, which makes it provides higher computational speeds, and it is more convenient for hardware implementation than the real number domain chaotic systems.  As far as we know, Topological mixing implies Devaney's definition of chaos in compact space, but not vise versa~\cite{Xiong:Chaos:ICTP1991,Zhou:Dynamics:1997}. Therefore, based on the fact that HDDCS satisfies Devaney's definition of chaos, this article will prove that the system is also topological mixing. Note that topological mixing is a stronger version of transitivity, and topological mixing implies transitivity.

The main construction method of HDDCS  in~\cite{Wang:hddcs:IEEECSI2016} is the experimental trial-and-error method. This method may be effective for modeling low-dimensional chaotic systems, but more and more limitations appear when the dimension increases. This article presents a more general design method for HDDCS via loop-state contraction algorithm.
The design of iterative functions is the starting point of this research, and then using simple bitwise operations to create the iterative function. In this way, the connectivity of the state transition diagram of $G_F$ in HDDCS will not change as finite precision changes. Then the set of all possible iterative functions $S$ is obtained by that. Note that the strong connectivity for $G_F$ in HDDCS is key to chaos. Therefore, the iterative function $F\in S$ is constructed according to the loop-state contraction algorithm, so that the state transition diagram of $G_F$ corresponding to $F$ is strongly connected. The loop is a non-empty directed path in the state transition diagram, where the first state and the last state are the same~\cite{Dharwadker:graphtheory2007}.
To contract the loop in a state transition diagram is to replace the set of states in the loop (and each occurrence of all the states in the set in any edge) by a single new state, and to delete any subsequent self-loops (edges that join a state to itself) and multi-edges~\cite{Khuller:sccnon1998}. Each edge in the resulting state transition diagram is identified with the corresponding edge in the original state transition diagram or, in the case of multiedges, the single remaining edge is identified with any one of the corresponding edges in the original state transition diagram.

Finally, the adjacency matrix corresponding to HDDCS is used to construct a chaotic ESN. ESN~\cite{Jaeger:ESN:2001} is a Recurrent Neural Network (RNN) proposed by Jaeger in 2001. Because of its recursive processing of historical information, it has short-term memory capabilities, and it is suitable for processing the information with strong correlation in time and space sequence, especially the ESN can guarantee fastness and global optimal, which avoids the inherent characteristics of the traditional RNN, such as easy to fall into local optimality and difficulty in stability analysis. Therefore, ESN shows strong application prospect in speech recognition~\cite{Haridas:ESN:2018}, language modeling~\cite{Morse:ESN:2017}, modeling and control of nonlinear system ~\cite{Yao:ESN:2018}, etc.  Especially in 2004, the prediction accuracy of Mackey-Glass time series  achieved a qualitative leap~\cite{Jaeger:ESN:2004,Li:Chaotic:ITNNLS2012,Ren:Performance:ITC2020}. As a time series with a chaotic attractor, the Mackey-Glass system has become one of the benchmark problems for time series prediction in both the neural network and fuzzy logic fields~\cite{Shi:MGS:Ieee2007}. In this article, the adjacency matrix corresponding to HDDCS is used to construct the network structure of the reservoir, which makes ESN be a chaotic ESN~\cite{Bahi:Neural:2012}.

The remainder of this article is organized as follows. The description of HDDCS and the proof of its topologically mixing are provided in Sec.~\ref{DCSStrong}. Section~\ref{Construction of HDDCS} presents the detail of the general design method of constructing iterative function via loop-state contraction algorithm. The application of the adjacency matrix corresponding to HDDCS to construct chaotic ESN for the Mackey-Glass time-series prediction is demonstrated in Sec.~\ref{Prediction}. The last section concludes the article.

\section{Proof of topological mixing of $G_F$}
\label{DCSStrong}
Topological mixing is a stronger property than transitivity~\cite{Zhou:Dynamics:1997}. This section proves that if the state transition diagram of $G_F$ in HDDCS is strongly connected, then $G_F$ is topological mixing in a compact space. Therefore, the description of $G_F$ in~\cite{Wang:hddcs:IEEECSI2016} is reviewed at first, then the concepts of metric space and compact space are  introduced, and then  $G_F\colon \mathcal{E}\to \mathcal{E}$ is proved to be  a continuous mapping in $(\mathcal{E}, d)$ (Theorem~\ref{thm1}). Finally, it is proved that $G_F$ is topological mixing in a compact space (Theorem~\ref{thm2}).

\subsection{Description of $G_F$ in HDDCS }
Let us recall the description of $G_F$ in HDDCS, as previously published in \cite{Wang:hddcs:IEEECSI2016}. Define $\mathcal{E}$ as the set of points $E$ of the form 
$((s, u, \ldots, v), (x_1, x_2, \ldots, x_m))$, where $s, u, \ldots, v$ are $m$ independent 
random sequences, while $x_1$, $x_2$, $\ldots$, $x_m$ are $m$ real numbers, each one is represented by $N$-bit floating numbers (limited accuracy), that is $N=P+Q$, $P$ is the number of binary digits for integral part and $Q$ is the number of binary digits for fractional part. Let us now define the mapping $G_F:\mathcal{E}\rightarrow\mathcal{E}$ as
\begin{equation}
	\begin{IEEEeqnarraybox}[][c]{lll}
	\IEEEstrut
	G_F(E) & = & G_F((s, u, \ldots, v), (x_1, x_2, \ldots, x_m))\\
	  & = & ((\sigma(s), \sigma(u), \ldots, \sigma(v)), (H_{F_1}(i(s), (x_1, x_2, \\
	  & & \ldots, x_m)), H_{F_2}(i(u), (x_1, x_2, \ldots, x_m)),\\ 
	  & & \ldots, H_{F_m}(i(v),(x_1, x_2, \ldots, x_m))))\,. 
	\IEEEstrut
	\end{IEEEeqnarraybox}
	\label{GF}
	\end{equation}	
According to Eq.\eqref{GF}, the general form of the iterative equation controlled by random sequences for HDDCS is
\begin{equation}
	E^{k+1}=G_F(E^k) (k=0,1,2,\ldots)\,.  
	\label{GF2}
	\end{equation} 
In Eq.\eqref{GF}, ${{H}_{{F_{j}}}}(i(u),({x_{1}},{x_2},\ldots ,{x_m}))(j=1,2,\ldots ,m)$ is
\begin{equation}
	\resizebox{1\hsize}{!}{$
	\left\{
	\begin{IEEEeqnarraybox}[][c]{lll}
	\IEEEstrut
	H_{F_1}(i(s), (x_1, x_2, \ldots, x_m)) &=& ((x_1\cdot\overline{i(s)})+(F_1(\cdot)\cdot i(s)))\,, \\
	H_{F_2}(i(u), (x_1, x_2, \ldots, x_m)) &=& ((x_2\cdot\overline{i(u)})+(F_2(\cdot)\cdot i(u)))\,, \\
	&\vdots& \\
	H_{F_m}(i(v), (x_1, x_2, \ldots, x_m)) &=& ((x_m\cdot\overline{i(v)})+(F_m(\cdot)\cdot i(v)))\,,
	\IEEEstrut
   	\end{IEEEeqnarraybox}
	\right.
	\label{BVN3}
	$}
\end{equation}        
where the operators $``\cdot"$, $``\overline{(\cdot)}"$, and $``+"$ denote bitwise AND, bitwise NOT (negation), and bitwise OR, respectively, $\sigma(w)$ ($w\in\{s, u, \cdots, v\}$) shifts one item in the one-sided infinite sequence $w={{w}^{1}}{{w}^2}\cdots {{w}^n}\cdots $ to the left, and at the $k$-th iteration,    
${{\sigma }^{k}}(w)={{w}^{k+1}}{{w}^{k+2}}\cdots {{w}^n},k=1,2,\ldots ,$
and $i(w)( w\in\{s, u, \cdots, v\})$ is equal to the overflow 
from the left shifting of the sequence $w$.

Note that in Eq.\eqref{BVN3}, ${{F}_{j}}(\cdot )(j=1,2,\ldots ,m)$ represents an iterative function, and its general form is
\begin{equation}
	\left\{
	\begin{IEEEeqnarraybox}[][c]{ll}
	\IEEEstrut
	F_1(\cdot) & = F_1(x_{1},x_{2},\ldots, x_{m})\,,\\
	F_2(\cdot) & = F_2(x_{1},x_{2},\ldots, x_{m})\,,\\
		  &\vdots \\
	F_m(\cdot) & = F_m(x_{1},x_{2},\ldots, x_{m})\,.
	\IEEEstrut
	\end{IEEEeqnarraybox} \right.
	\end{equation}
How to select the appropriate iterative function ${F_{1}},{F_2},\ldots ,{F_m}$ to make the state transition diagram of ${G_F}$ in HDDCS represented by Eq.\eqref{GF2} is strongly connected is an important content of this article.

\subsection{Metric space}
In mathematics, a distance function on a given set $\mathcal{E}$ is a function $d:\mathcal{E}\times \mathcal{E} \rightarrow R$, where $R$ denotes the set of real numbers, that satisfies the following conditions:
\begin{enumerate}
    \item Non-negativity and identity of indiscernibles: $d(E, \hat{E})\geq 0$, and  $d(E, \hat{E})=0$ iff $E=\hat{E}$. 
    \item Symmetry: $d(E, \hat{E})=d(\hat{E}, E)$. 
    \item Triangle inequality: $d(E, \hat{E})\leq d(E, \tilde{E})+d(\tilde{E}, \hat{E})$. 
    \end{enumerate}
It should be noted that the distance $d$ cannot be infinite, otherwise the above three properties cannot be satisfied.

According to~\cite{Wang:hddcs:IEEECSI2016},  the distance for HDDCS in the metric space is defined as
\begin{equation}
	\label{distance}
	\begin{IEEEeqnarraybox}[][c]{rcl}
	\IEEEstrut
	d(E, \hat{E}) & = & d(((s, u, \ldots, v), (x_1, x_2, \ldots, x_m)),\\
	&&((\hat{s}, \hat{u}, \ldots, \hat{v}), (\hat{x}_1, \hat{x}_2, \ldots, \hat{x}_m)))\\
	&=&d_s(s, \hat{s})+d_u(u, \hat{u})+\cdots
	+d_v(v, \hat{v})\\
	&&+d_x((x_1, x_2, \ldots, x_m), (\hat{x}_1, \hat{x}_2, \ldots, \hat{x}_m))\,,
	\IEEEstrut
	\end{IEEEeqnarraybox} 
	\end{equation}
where
\begin{equation*}
	\left\{
	\begin{IEEEeqnarraybox}[][c]{lll}
	\IEEEstrut
		d_s(s, \hat{s})&=&\sum_{k=1}^\infty \frac{|s^k-\hat{s}^k|}{2^{Nk}}\,, \\
		d_u(u, \hat{u})&=&\sum_{k=1}^\infty \frac{|u^k-\hat{u}^k|}{2^{Nk}}\,, \\
			&\vdots&\\
		d_v(v, \hat{v})&=&\sum_{k=1}^\infty \frac{|v^k-\hat{v}^k|}{2^{Nk}}\,, 
	\IEEEstrut
	\end{IEEEeqnarraybox}%
	\right.
	\end{equation*}
\begin{equation*}
	\label{distancex}
	\begin{IEEEeqnarraybox}[][c]{rl}
		&d_{x}((x_1, x_2, \ldots, x_m), (\hat{x}_1, \hat{x}_2, \ldots, \hat{x}_m ))\\
		= &\sqrt{(x_1-\hat{x}_1)^2+(x_2-\hat{x}_2)^2+\cdots+(x_m-\hat{x}_m)^2}, 
	\IEEEstrut
	\IEEEstrut
\end{IEEEeqnarraybox} 
\end{equation*}
$x_1,\hat{x}_1,x_2,\hat{x}_2,\ldots,x_m,\hat{x}_m$ are binary forms of real numbers 
with $N$-bit finite precision, and $0\leq d_{x}\leq\sqrt{m}(2^P-2^{-Q})$.

In Eq.\eqref{distance}, the general expression of two different $m$ one-sided infinite random sequences as
\begin{equation*}
	\left\{
		\begin{IEEEeqnarraybox}[][c]{ll}
		\IEEEstrut
			s & = s^1s^2\cdots s^n\cdots\\
			u & = u^1u^2\cdots u^n\cdots\\
			  & \vdots\\
			v & = v^1v^2\cdots v^n\cdots
		\IEEEstrut
		\end{IEEEeqnarraybox}
		\right.\,,
	\left\{
			\begin{IEEEeqnarraybox}[][c]{ll}
			\IEEEstrut
				\hat{s} & = \hat{s}^1\hat{s}^2\cdots \hat{s}^n\cdots\\
				\hat{u} & = \hat{u}^1\hat{u}^2\cdots \hat{u}^n\cdots\\
				  & \vdots\\
				\hat{v} & = \hat{v}^1\hat{v}^2\cdots \hat{v}^n\cdots
			\IEEEstrut
			\end{IEEEeqnarraybox}
			\right.\,.
\end{equation*}			

Obviously, the above Eq.\eqref{distance} satisfies the non-negativity and identity of indiscernible, the symmetry and triangular inequality properties, so this is a distance that makes $(\mathcal{E}, d)$ a metric space. Note that the definition of distance is not unique. As long as it satisfies the three properties and is not infinite, it can be used as the definition of distance. Choosing an appropriate definition of distance is more conducive to the study of the problem.

\subsection{Continuity of ${G_F}$ and compact space}
\begin{lemma}
	\label{lemma1}
	Let $w\in\{s, u, \cdots, v\}$, $w=w^1w^2w^3 \ldots w^n \ldots$ and $\hat{w} = \hat{w}^1\hat{w}^2\hat{w}^3 \ldots \hat{w}^n \ldots$,
	the metric distance $d$ satisfies that if
	$w^i = \hat{w}^i$ for $i=1,2,3,\ldots n$, then
	$d(w,\hat{w})\le 1/{2^{Nn}}$. If $d(w,\hat{w})\le 1/{2^{Nn}}$, there must be $w^i = \hat{w}^i$ for $i=1,2,3,\ldots n$.
	And $w^k, \hat{w}^k \in [0,2^P-2^{-Q}]$ for $k\in \mathbb{Z}^+$.
	\end{lemma}
The lemma can let us quickly determine whether the two sequences are close to each other.
From intuitive observation, we can assure two sequences are close to each other as long as they have a considerable number of consistent foregoing entries.

\begin{theorem}
	\label{thm1}
$G_F:\mathcal{E}\rightarrow\mathcal{E}$ is a continuous function in $(\mathcal{E},d)$.
\end{theorem}
\begin{proof}
Before the proof of topological mixing of ${G_F}$, let's first review the definition of continuity. ${G_F}$ is continuous in $(\mathcal{E},d)$. In other words, for all $\varepsilon>0$, there always exists $\delta>0$ such that $d(E, \hat{E})<\delta$ with $E\neq\hat{E}$, implies that $d(G_F(E), G_F(\hat{E}))<\varepsilon$, the function $G_F$ is consecutive at $\hat{E}$.
\begin{enumerate} 
\item According to the definition of continuity, $\forall \varepsilon >0$, a positive integer ${k_0}=\left\lfloor ({{\log }_2}m-{{\log }_2}\varepsilon )/N \right\rfloor +1$ can always be found, so that $m/{2^{N{k_0}}}<\varepsilon$ holds, where $m$ is the dimension, $N=P+Q$ is the finite precision, and the operator ``$\left \lfloor {} \right\rfloor $" means round down.
\item According to the found ${k_0}$, $\exists \delta =m/{2^{N({k_0}+1)}}$, from Eq.\eqref{distance}, $d(E,\hat{E})<\delta$ is obtained as
	\begin{equation}
		\begin{IEEEeqnarraybox}[][c]{lll}
		\IEEEstrut
		d(E, \hat{E}) & = & d_s(s, \hat{s})+d_u(u, \hat{u})+\cdots+d_v(v, \hat{v})+\\
		  & & d_x((x_1, x_2, \ldots, x_m), (\hat{x}_1, \hat{x}_2, \ldots, \hat{x}_m)) \\
		  & <&m/{2^{N({k_0+1})}}=\delta\,, 
		\IEEEstrut
		\end{IEEEeqnarraybox}
		\label{distance2}
		\end{equation}	
By making Eq.\eqref{distance2} hold, we must set
$$({x_{1}},{x_2},\ldots ,{x_m})=({\hat{x}_{1}},{\hat{x}_2},\ldots ,{\hat{x}_m}),$$
otherwise ${d_x}(({x_{1}},{x_2},\ldots ,{x_m}),({\hat{x}_{1}},{\hat{x}_2},\ldots ,{\hat{x}_m}))\ge {2^{-Q}}$, but considering that the arbitrary value ${k_0}=\left\lfloor ({{\log }_2}m-{{\log }_2}\varepsilon )/N \right\rfloor +1$, may make $d(E,\hat{E})\ge {2^{-Q}}\ge m/{2^{N({k_0}+1)}}$, which contradicts to Eq.\eqref{distance2}.

\item We must set $({x_{1}},{x_2},\ldots ,{x_m})=({\hat{x}_{1}},{\hat{x}_2},\ldots ,{\hat{x}_m})$ to obtain ${d_x}(({x_{1}},{x_2},\ldots ,{x_m}),({\hat{x}_{1}},{\hat{x}_2},\ldots ,{\hat{x}_m}))=0$. According to Eq.\eqref{distance2}, one can obtain
\begin{equation}
	\begin{IEEEeqnarraybox}[][c]{lll}
	\IEEEstrut
	d(E, \hat{E}) & = & d_s(s, \hat{s})+d_u(u, \hat{u})+\cdots+d_v(v, \hat{v})\\
	  & <&m/{2^{N({k_0+1})}}=\delta\,, 
	\IEEEstrut
	\end{IEEEeqnarraybox}
	\label{distance8}
	\end{equation}	
By making Eq.\eqref{distance8} hold, we should ensure that the distances of $m$ pairs of independent random sequences satisfy
\begin{equation}
	\label{eq8}
	\left\{
	\begin{IEEEeqnarraybox}[][c]{lll}
	\IEEEstrut
		d_s(s, \hat{s})&\leq&1/{2^{N({k_0+2})}}\,, \\
		d_u(u, \hat{u})&\leq&1/{2^{N({k_0+2})}}\,, \\
			&\vdots&\\
		d_v(v, \hat{v})&\leq&1/{2^{N({k_0+2})}}\,, 
	\IEEEstrut
	\end{IEEEeqnarraybox}%
	\right.
	\end{equation}                                    

Referring to the Lemma \ref{lemma1}, if ${d_s}(s,\hat{s})\le 1/{2^{N({k_0}+2)}}$, there must be ${s^i}={{ \hat{s}}^i}(i=1,2,\ldots ,{k_0}+2)$. Similar results can be obtained as ${d_{u}}(u,\hat{u})\le 1/{2^{N({k_0}+2)}}$, $\ldots$, ${d_{v}}(v,\hat{v})\le 1/{2^{N({k_0}+2)}}$, then ${{u}^i}={{ \hat{u}}^i}(i=1,2,\ldots ,{k_0}+2)$, $\ldots$, ${{v}^i }={{\hat{v}}^i}(i=1,2,\ldots ,{k_0}+2)$. So, we can obtain
\begin{equation*}
	\begin{IEEEeqnarraybox}[][c]{lll}
	\IEEEstrut
	d(E, \hat{E}) & = & d_s(s, \hat{s})+d_u(u, \hat{u})+\cdots+d_v(v, \hat{v})\\
	  & \leq&m/{2^{N({k_0+2})}}<m/{2^{N({k_0+1})}}=\delta .
	\IEEEstrut
	\end{IEEEeqnarraybox}
	\label{distance7}
	\end{equation*}	

\item The following further proves that when $d(E,\hat{E})<\delta $, the inequality $d({G_F}(E),{G_F}(\hat{E}))<\varepsilon $ always holds. In fact, according to the definition of $d({G_F}(E),{G_F}(\hat{E}))$, we have
\begin{equation}
	\label{eq9}
	\begin{IEEEeqnarraybox}[\IEEEeqnarraystrutmode\IEEEeqnarraystrutsizeadd{2pt}{2pt}][c]{rCl}
	\IEEEeqnarraymulticol{3}{l}{d(G_F(E), G_F(\hat{E})\!)}\\
	& = & d_s(\sigma(s), \sigma(\hat{s})\!)+d_u(\sigma(u),\sigma(\hat{u})\!)+ \cdots   \\
	&&  d_v(\sigma(v),\sigma(\hat{v})\!)\!+\!d_x(\!(H_{F_1}(i(s),(x_1, x_2, \ldots, x_m)\!), \\
	&&	H_{F_2}(i(u),(x_1, x_2, \ldots, x_m)\!), \ldots,\\
	&&H_{F_m}(i(v), (x_1, x_2, \ldots, x_m)\!)\!), \\
	&&(H_{F_1}(i(\hat{s}),(\hat{x}_1, \hat{x}_2, \ldots, \hat{x}_m)\!),\\
	&& H_{F_2}(i(\hat{u}), (\hat{x}_1, \hat{x}_2, \ldots, \hat{x}_m)\!), \ldots,\\
	&&H_{F_m}(i(\hat{v}), (\hat{x}_1, \hat{x}_2, \ldots, \hat{x}_m)\!)\!)\!),		
	\end{IEEEeqnarraybox}	
	\end{equation}
where $\sigma (\cdot )$ is the operation of shifting the sequence one place to the left, so the first ${k_0}+1$ elements of $\sigma (s)$ and  $\sigma (\hat{s})$ are the same, then ${d_s}(\sigma (s),\sigma (\hat{s}))\le 1/{2^{N({k_0}+1)}}$ from Lemma ~\ref{lemma1}. Also, first ${k_0}+1$ elements of $\sigma (u), \ldots, \sigma (v)$ and $\sigma (\tilde{u}), \ldots,\sigma (\tilde{v})$ are the same and ${d_{u}}(\sigma (u),\sigma (\hat{u}))\le 1/{2^{N({k_0}+1)}},\ldots ,{d_{v}}(\sigma (v),\sigma (\hat{v}))\le 1/{2^{N({k_0}+1)}}$, which makes the inequality
\begin{equation}
	\begin{IEEEeqnarraybox}[][c]{ll}
	\IEEEstrut
	&{d_s}(\!\sigma (s),\sigma (\hat{s})\!)\!+\!{d_{u}}(\!\sigma (u),\sigma (\hat{u})\!)\!+\!\cdots \!+\!{d_{v}}(\!\sigma (v),\sigma (\hat{v})\!)\\
	\le  & m/{2^{N({k_0}+1)}}<m/{2^{N{k_0}}}<\varepsilon 	
	\IEEEstrut
	\end{IEEEeqnarraybox}
	\label{eq10}
	\end{equation}	
hold.
\item In Eq.\eqref{eq9}, notice that $i(\cdot )$ is equal  to  the  overflow from the left shifting of the sequence, that is, $i(s)={s^{1}}, i(u)={{u}^{1}},\ldots ,i(v)={{v}^{1}}$. Because the first ${k_0}+2$ elements of $s,u,\ldots,v$ and $\hat{s},\hat{u},\ldots,\hat{v}$ are exactly the same, we know that ${s^{1}}={{\hat{s}}^{1}},{{u}^{1}}={{\hat{u}}^{ 1}},\ldots ,{{v}^{1}}={{\hat{v}}^{1}}$,  ${x_{1}}={\hat{x}_{1 }},{x_2}={\hat{x}_2},\ldots ,{x_m}={\hat{x}_m}$, then, one has
\end{enumerate}
\begin{equation}
	\left\{
	\begin{IEEEeqnarraybox}[][c]{lll}
	\IEEEstrut
	H_{F_1}(s^1, (x_1, x_2, \ldots, x_m)\!)& = &H_{F_1}(\hat{s}^1, (\hat{x}_1, \hat{x}_2, \ldots, \hat{x}_m)\!), \\
	H_{F_2}(u^1, (x_1, x_2, \ldots, x_m)\!)& = &H_{F_2}(\hat{u}^1, (\hat{x}_1, \hat{x}_2, \ldots, \hat{x}_m)\!), \\
			 & \vdots &\\
	H_{F_m}\!(v^1, (x_1, x_2, \ldots, x_m)\!)&=&H_{F_m}\!(\hat{v}^1, (\hat{x}_1, \hat{x}_2, \ldots, \hat{x}_m)\!).
	\IEEEstrut
	\end{IEEEeqnarraybox}
	\right.
	\label{eq11}
	\end{equation}	

    According to Eq.~\eqref{eq11}, 
	\begin{IEEEeqnarray*}{lcl}
		d_x((H_{F_1}(i(s), (x_1, x_2, \ldots, x_m)),&& \\ 		
		H_{F_2}(i(u),(x_1, x_2, \ldots, x_m)), \ldots,&&\\	
		H_{F_m}(i(v), (x_1, x_2, \ldots, x_m))),&&\\		
		(H_{F_1}(i(\hat{s}),(\hat{x}_1, \hat{x}_2, \ldots, \hat{x}_m)), &&\\
		H_{F_2}(i(\hat{u}), (\hat{x}_1, \hat{x}_2, \ldots, \hat{x}_m)), \ldots,&&\\			
			H_{F_m}(i(\hat{v}), (\hat{x}_1, \hat{x}_2, \ldots, \hat{x}_m))))&=&0 \,. 
		\IEEEyesnumber\label{eq12}
		\end{IEEEeqnarray*}

Substituting Eq.~\eqref{eq10} and Eq.~\eqref{eq12} into Eq.~\eqref{eq9}, we get
\begin{IEEEeqnarray}{r}
	\quad d({G_F}(E),{G_F}(\hat{E}))\le m/{2^{N({k_0}+1)}}<m/{2^{N{k_0}}}<\varepsilon \,,	\nonumber\\*	
\end{IEEEeqnarray}                           
which shows that ${G_F}$ is continuous.

\end{proof}
In summary, it can be seen that ${G_F}:\mathcal{E}\to \mathcal{E}$ is a continuous mapping in the metric space $(\mathcal{E},d)$, and it also means the metric space $(\mathcal{E},d)$ is a compact space.

\subsection{State transition diagram and strong connectivity}
Given a digital chaotic system, a state and its interval is mapped to another one. Considering the mapping relation as a directed edge (link), the state transition  diagram of the chaotic system can be build up. As shown in \cite{Hsu:CellMap:IJBC92,Shreim:NetworkCA:2007}, the associated state transition  diagram can demonstrate some dynamical properties of digital chaotic systems that cannot be observed by the previous analytic methods. For the state transition  diagram of $G_F$ in HDDCS, all the possible combinations of $(x_1,x_2,\ldots,x_m)$ are the states, and there is a directed edge from state $(\hat{x}_1,\hat{x}_2,\ldots,\hat{x}_m)$ to another state $(\tilde{x}_1,\tilde{x}_2,\ldots,\tilde{x}_m)$ if
\begin{equation}
    \label{eq14}
    \begin{IEEEeqnarraybox}[\IEEEeqnarraystrutmode\IEEEeqnarraystrutsizeadd{2pt}{2pt}][c]{rCl}
    &&(G_F((\hat{s}, \hat{u}, \ldots, \hat{v}), (\hat{x}_1, \hat{x}_2, \ldots, \hat{x}_m)))_{x_1, x_2, \ldots, x_m}\\
    &=&(\tilde{x}_1, \tilde{x}_2, \ldots, \tilde{x}_m)\,,
    \end{IEEEeqnarraybox}
    \end{equation}

\begin{definition}[\cite{Wang:hddcs:IEEECSI2016}]
	\label{strong connectivity}
    If each state in the state transition diagram can reach any other state through a directed edge, the state transition  diagram is strongly connected.
\end{definition}

\subsection{Topological mixing of $G_F$}

Topological mixing, Li Yorke and Devaney's chaos are three well known and common criterias of chaos in a discrete dynamical system. The relationship between them is as follows: topological mixing implies both Li Yorke and Devaney's chaos in compact spaces~\cite{Xiong:Chaos:ICTP1991,Zhou:Dynamics:1997,Li:Period:2004}, but  not vise versa.

\begin{theorem}
	\label{thm2}
If the state transition diagram of ${G_F}$ is strongly connected, then ${G_F}$ is topological mixing in the compact space $(\mathcal{E},d)$.
\end{theorem}

\begin{proof}
Before the proof, first review the topological mixing of ${G_F}$ in the compact space $(\mathcal{E},d)$. The so-called topological mixing specifically refers to that, for the non-empty open set $U,V\subset \mathcal{E}$, there is a positive integer ${n_0}$, which satisfies $G_F^n(U)\cap V\ne \varnothing ,\forall n\geq{n_0}$.

Suppose that the non-empty open set $U$ is centered on $((s,u,\ldots ,v),(x_{1},x_2,\ldots ,x_m))$ with a radius $r$, notice the center point $((s,u,\ldots ,v),({x_{1}},{x_2},\ldots ,{x_m}))\in U\subset \mathcal{E}$ can be expressed as
\begin{equation}
	\begin{IEEEeqnarraybox}[\IEEEeqnarraystrutmode\IEEEeqnarraystrutsizeadd{2pt}{2pt}][c]{rCl}
	\IEEEeqnarraymulticol{3}{l}{
	((s, u, \ldots, v), (x_1, x_2, \ldots, x_m))}\\
	\quad & = & (((s^1s^2\cdots s^{k_0}\cdots s^n\cdots ), \\
	\quad & & (u^1u^2\cdots u^{k_0}\cdots u^n\cdots ), \ldots, \\
	\quad & & (v^1v^2\cdots v^{k_0}\cdots v^n\cdots )), \\
	\quad & & (x_1, x_2, \ldots, x_m))\,.
	\end{IEEEeqnarraybox}
	\end{equation}
First, we need to prove that for any point $((s'u',\ldots ,v'),(x{{'}_{1}},x{{'} _2},\ldots ,x{{'}_m}))\in \mathcal{E}$, we can find $((\tilde{s},\tilde{u},\ldots ,\tilde{v}),({\tilde{x}_{1}},{\tilde{x}_2},\ldots ,{\tilde{x}_m}))\in U\subset \mathcal{E}$   that can reach the point $((s'u ',\ldots ,v'),(x{{'}_{1}},x{{'}_2},\ldots ,x{{'}_m}))\in \mathcal{E}$ after ${n_0}={k_0}+{2^{Nm}}$-th iteration, where ${k_0}=\left\lfloor ({{\log }_2}m-{{\log }_2}r)/N \right\rfloor +1$, $N$ is the finite precision of ${G_F}$, $m$ is the dimension of ${G_F}$. The proof process is as follows:
\begin{enumerate}
\item Given the spherical radius $r$ of $U$ is less than ${2^{-Q}}$, if $({x_{1}},{x_2},\ldots ,{x_m})$ and $({\tilde {x}_{1}},{\tilde{x}_2},\ldots ,{\tilde{x}_m})$ do not coincide in the $m$-dimensional space, we can obtain $(\tilde{x}_1,\tilde{x}_2,\ldots,\tilde{x}_m)\neq(x_1,x_2,\ldots,x_m)$ such that
\begin{equation*}
\begin{IEEEeqnarraybox}[][c]{lll}
&&{d_x}(({\tilde{x}_{1}},{\tilde{x}_2},\ldots ,{\tilde{x}_m}),({x_{1}},{x_2},\ldots ,{x_m}))\\
&=& \sqrt{{{({\tilde{x}_{1}}-{x_{1}})}^2}+{{({\tilde{x}_2}-{x_2})}^2}+\cdots +{{({\tilde{x}_m}-{x_m})}^2}}\\
&\geq& {2^{-Q}}>r\,,
\IEEEstrut
\end{IEEEeqnarraybox}
\end{equation*}	
Then, one has $((\tilde{s},\tilde{u},\ldots ,\tilde{v}),({\tilde{x}_{1}},{\tilde{x}_2},\ldots ,{\tilde{x}_m}))\notin U$. Therefore, $({\tilde{x}_{1}},{\tilde{x}_2},\ldots ,{\tilde{x}_m})=({x_{1}},{x_2},\ldots ,{x_m})$ must first be satisfied.
\item According to Lemma~\ref{lemma1}, if the first ${k_0}$ elements of $\tilde{s},\tilde{u},\ldots ,\tilde{v}$ and $s,u,\ldots ,v$ are the same, then
${d_s}(s,\tilde{s})\le 1/{2^{N{k_0}}},{d_{u}}(u,\tilde{u})\le 1/{2^ {N{k_0}}},\ldots ,{d_{v}}(v,\tilde{v})\le 1/{2^{N{k_0}}}$
. So $\forall \ r<1$, an integer $k_0$ satisfying the relation 
$${d_s}(s,\tilde{s})+{d_{u}}(u,\tilde{u})+\cdots +{d_{v}}(v,\tilde{v})\le m/{2^{N{k_0}}}<r$$ can always be found.
\item If after the $k_0$-th iteration, equality
	\begin{equation}
		\label{eq16}
		\begin{IEEEeqnarraybox}[][c]{ll}
			&(G_F^{k_0}((\tilde{s}, \tilde{u}, \ldots, \tilde{v}), (\tilde{x}_1, \tilde{x}_2, \ldots, \tilde{x}_m)))_{x_1, x_2, \ldots, x_m}\\
			=&(x_1', x_2', \ldots, x_m')\,,
		\IEEEstrut
		\end{IEEEeqnarraybox}
		\end{equation}
exists, we know that $({\tilde{x}_{1}},{\tilde{x}_2},\ldots ,{\tilde{x}_m})=({x_{1}},{x_2} ,\ldots ,{x_m})$, the point is found in $U$ as	
\begin{equation}
	\begin{IEEEeqnarraybox}{rCl}
	&&	((\tilde{s}, \tilde{u}, \ldots, \tilde{v}), (\tilde{x}_1, \tilde{x}_2, \ldots, \tilde{x}_m))\\
	 & = & (((s^1s^2\cdots s^{k_0}
	\check{s}^{k_0+1}\check{s}^{k_0+2}\cdots\check{s}^{n_0}{s'}^1{s'}^2\cdots{s'}^n\cdots),\\
	 & & (u^1u^2\cdots u^{k_0}\check{u}^{k_0+1}\check{u}^{k_0+2}\cdots\check{u}^{n_0}{u'}^1{u'}^2\cdots{u'}^n\cdots),  \\
	 & & \ldots,\\
	 & & (v^1v^2\cdots v^{k_0}\check{v}^{k_0+1}\check{v}^{k_0+2}\cdots\check{v}^{n_0}{v'}^1{v'}^2\cdots{v'}^n\cdots)), \\
	 & & (\tilde{x}_1, \tilde{x}_2, \ldots, \tilde{x}_m))\\
	 &\in &U\,,
	\end{IEEEeqnarraybox}
	\end{equation}
where ${{\check{s}}^{{k_0}+1}}{{\check{s}}^{{k_0}+2}}\cdots {{\check{s}}^{{n_0}}}$, ${{\check{u}}^{{k_0}+1}}{{\check{u}}^{{k_0}+2}}\cdots {{\check{u}}^{{n_0}}}$, $\ldots$, ${{\check{v}}^{{k_0}+1}}{{\check{v}}^{{k_0}+2}}\cdots {{\check{v}}^{{n_0}}}$  are all 0s.
Note that the purpose of setting ${{\check{s}}^{{k_0}+1}}{{\check{s}}^{{k_0}+2}}\cdots {{\check{s}}^{{n_0}}}$, ${{\check{u}}^{{k_0}+1}}{{\check{u}}^{{k_0}+2}}\cdots {{\check{u}}^{{n_0}}}$, $\ldots$, ${{\check{v}}^{{k_0}+1}}{{\check{v}}^{{k_0}+2}}\cdots {{\check{v}}^{{n_0}}}$ all 0s is to ensure that the point $((\tilde{s},\tilde{u},\ldots ,\tilde{v}),({\tilde{x}_{1}},{\tilde{x}_2},\ldots ,{\tilde{x}_m}))\in U$ can reach any point $((s'u',\ldots ,v'),(x{{'}_{1}},x{{'}_2},\ldots ,x{{'}_m}))\in \mathcal{E}$ at the same time after ${n_0}={k_0}+{2^{Nm}}$-th iteration. By making the equation
	\begin{IEEEeqnarray*}{l}
		\label{eq18}
		G_F^{n_0}(\!(\!(s^1s^2\cdots s^{k_0}\check{s}^{k_0+1}\check{s}^{k_0+2}\cdots\check{s}^{n_0}{s'}^1{s'}^2\cdots {s'}^n\cdots), \\ 
		(u^1u^2\cdots u^{k_0}\check{u}^{k_0+1}\check{u}^{k_0+2}\cdots\check{u}^{n_0}{u'}^1{u'}^2\cdots {u'}^n\cdots),\\
		\ldots,\\		
		(v^1v^2\cdots  v^{k_0}\check{v}^{k_0+1}\check{v}^{k_0+2}\cdots\check{v}^{n_0}{v'}^1 {v'}^2\cdots{v'}^n\cdots)\!),  \\
		(\tilde{x}_1, \tilde{x}_2, \ldots, \tilde{x}_m)\!)\\	
		=((s', u', \ldots, v'), (x_1', x_2', \ldots, x_m'))\,,
		\IEEEyesnumber
		\end{IEEEeqnarray*}
hold, we certainly can find the point $((\tilde{s},\tilde{u},\ldots ,\tilde{v}),({\tilde{x}_{1}} ,{\tilde{x}_2},\ldots ,{\tilde{x}_m})$ in $U$, which can reach any point $( (s'u',\ldots ,v'),(x{{'}_{1}},x{{'}_2},\ldots ,x{{'}_m}))$ in $\mathcal{E}$ after ${n_0}$-th iteration.
\item If after ${k_0}$-th iteration, inequality 
\begin{IEEEeqnarray*}{ll}
	&(G_F^{k_0}((\tilde{s}, \tilde{u}, \ldots, \tilde{v}), (\tilde{x}_1, \tilde{x}_2, \ldots, \tilde{x}_m)))_{x_1, x_2, \ldots, x_m}\\
	\neq&(x_1', x_2', \ldots, x_m')\,,
\end{IEEEeqnarray*}
holds, set
\begin{IEEEeqnarray*}{ll}
	&(G_F^{k_0}((\tilde{s}, \tilde{u}, \ldots, \tilde{v}), (\tilde{x}_1, \tilde{x}_2, \ldots, \tilde{x}_m)))_{x_1, x_2, \ldots, x_m}\\
	=&(x_1'', x_2'', \ldots, x_m'')\,.
\end{IEEEeqnarray*}
It should be noted that the state  transition  diagram of ${G_F}$ is strongly connected, so there is at least one path between  $(x'{{'}_{1}},x'{{'}_2},\ldots ,x'{{' }_m})$ and $(x{{'}_{1}},x{{'}_2},\ldots ,x{{'}_m})$,  from ${k_0}+1$-th iteration and then iterate ${i_0}$ times
(${i_0}$ is equivalent to the number of edges traversed by the connected path between $(x'{{'}_{1}},x'{{'}_2},\ldots ,x'{{'}_m})$ and $(x{{ '}_{1}},x{{'}_2},\ldots ,x{{'}_m})$. 
Since the maximum number of states in the state  transition  diagram of ${G_F}$ is ${2^{Nm}}$, and the number of edges traversed by the longest path is ${2^{Nm}}-1$. Since $i_0 \leq {2^{Nm}}-1$, it satisfies ${k_0}+{i_0}<{n_0}$.
And its iterative process is controlled by ${{\hat{s}}^{{k_0}+j}},{{\hat{u}}^{{k_0}+j}},\ldots ,{{\hat {v}}^{{k_0}+j}}(j=1,2,\ldots ,{i_0})$, so that
\begin{equation}
	\begin{IEEEeqnarraybox}[][c]{ll}
		\IEEEstrut
	&G_F^{k_0+i_0}((\tilde{s}, \tilde{u}, \ldots, \tilde{v}), (\tilde{x}_1, \tilde{x}_2, \ldots, \tilde{x}_m))\\
	 =& ((s', u', \ldots, v'), (x_1', x_2', \ldots, x_m'))\,.
	 \IEEEstrut
\end{IEEEeqnarraybox}
\end{equation}
holds. Considering $({\tilde{x}_{1}},{\tilde{x}_2},\ldots ,{\tilde{x}_m})=({x_{1}},{x_2}, \ldots ,{x_m})$, the point is found in $U$ as
\begin{equation}
	\begin{IEEEeqnarraybox}[\IEEEeqnarraystrutmode\IEEEeqnarraystrutsizeadd{2pt}{2pt}][c]{Cl}
	&((\tilde{s}, \tilde{u}, \ldots, \tilde{v}), (\tilde{x}_1, \tilde{x}_2, \ldots, \tilde{x}_m))\\
	 = & (((s^1 s^2\cdots s^{k_0}\hat{s}^{k_0+1}\hat{s}^{k_0+2}\cdots \hat{s}^{k_0+i_0}{{\check{s}}^{{k_0}+{i_0}+1}} \\
	 &{{\check{s}}^{{k_0}+{i_0}+2}}\cdots {{\check{s}}^{{n_0}}}s'^1s'^2\cdots s'^n\cdots),  \\
	 &(u^1 u^2\cdots u^{k_0}\hat{u}^{k_0+1}\hat{u}^{k_0+2}\cdots \hat{u}^{k_0+i_0}{{\check{u}}^{{k_0}+{i_0}+1}}\\
	 &{{\check{u}}^{{k_0}+{i_0}+2}}\cdots {{\check{u}}^{{n_0}}} u'^1u'^2\cdots u'^n\cdots), \\
	 &\ldots, \\
	 &(v^1 v^2\cdots v^{k_0} \hat{v}^{k_0+1}\hat{v}^{k_0+2}\cdots \hat{v}^{k_0+i_0} {{\check{v}}^{{k_0}+{i_0}+1}}\\
	 &{{\check{v}}^{{k_0}+{i_0}+2}}\cdots {{\check{v}}^{{n_0}}}v'^1v'^2\cdots v'^n\cdots)),\\
	 &(\tilde{x}_1, \tilde{x}_2, \ldots, \tilde{x}_m))\\
	 \in & U\,, 
	\end{IEEEeqnarraybox}
	\end{equation}	
where ${{\check{s}}^{{k_0}+{i_0}+1}}{{\check{s}}^{{k_0}+{i_0}+2}}\cdots {{\check{s}}^{{n_0}}},{{\check{u}}^{{k_0}+{i_0}+1}}{{\check{u}}^{{k_0}+{i_0}+2}}\cdots \allowbreak {{\check{u}}^{{n_0}}},\ldots,{{\check{v}}^{{k_0}+{i_0}+1}}{{\check{v}}^{{k_0}+{i_0}+2}}\cdots {{\check{v}}^{{n_0}}}$ are all 0s.

Similarly, the purpose of setting ${{\check{s}}^{{k_0}+{i_0}+1}}{{\check{s}}^{{k_0}+{i_0}+2}}\cdots {{\check{s}}^{{n_0}}},{{\check{u}}^{{k_0}+{i_0}+1}}{{\check{u}}^{{k_0}+{i_0}+2}}\cdots {{\check{u}}^{{n_0}}},\ldots ,\allowbreak{{\check{v}}^{{k_0}+{i_0}+1}}{{\check{v}}^{{k_0}+{i_0}+2}}\cdots {{\check{v}}^{{n_0}}}$ all 0s is to ensure that the point $((\tilde{s},\tilde{u},\ldots ,\tilde{v}),({\tilde{x}_{1}},{\tilde{x}_2},\ldots ,{\tilde{x}_m}))\in U$ can reach any point $((s'u',\ldots ,v'),(x{{'}_{1}},x{{'}_2},\ldots ,x{{'}_m}))\in \mathcal{E}$ at the same time after ${n_0}$-th iteration, and 
\begin{IEEEeqnarray*}{ll}
	\label{eq21}
	&G_F^{n_0}(((s^1 s^2\cdots s^{k_0}\hat{s}^{k_0+1}\hat{s}^{k_0+2}\cdots \hat{s}^{k_0+i_0}{{\check{s}}^{{k_0}+{i_0}+1}} \\ 
	&{{\check{s}}^{{k_0}+{i_0}+2}}\cdots {{\check{s}}^{{n_0}}}s'^1s'^2\cdots s'^n\cdots),\\
	&(u^1 u^2\cdots u^{k_0}\hat{u}^{k_0+1}\hat{u}^{k_0+2}\cdots \hat{u}^{k_0+i_0}{{\check{u}}^{{k_0}+{i_0}+1}}\\		
	&{{\check{u}}^{{k_0}+{i_0}+2}}
	 \cdots {{\check{u}}^{{n_0}}} u'^1u'^2\cdots u'^n\cdots),   \\
	 &\ldots,\\		
	&(v^1 v^2\cdots v^{k_0} \hat{v}^{k_0+1}\hat{v}^{k_0+2}\cdots \hat{v}^{k_0+i_0}{{\check{v}}^{{k_0}+{i_0}+1}}\\
	&{{\check{v}}^{{k_0}+{i_0}+2}}\cdots {{\check{v}}^{{n_0}}}v'^1v'^2\cdots v'^n\cdots)),\\
	&(\tilde{x}_1, \tilde{x}_2, \ldots, \tilde{x}_m))\\
	=&((s', u', \ldots, v'), (x_1', x_2', \ldots, x_m'))\,. 
	\IEEEyesnumber
	\end{IEEEeqnarray*}

\end{enumerate}
According to Eq.\eqref{eq18} and \eqref{eq21}, and considering the arbitrariness of the point $((s'u',\ldots ,v'),(x{{'}_{1}},x{{'}_2},\ldots ,x{{'}_m}))$ in $\mathcal{E}$, we know
\begin{equation}
	\label{eq22}
G_F^{{n_0}}(U)=\mathcal{E}   
\end{equation}                                                  
holds.
On the basis of the above, ${G_F}$ is topologically mixing in the metric space $(\mathcal{E},d)$. The proof is as follows:

Although ${G_F}:\mathcal{E}\to \mathcal{E}$ is a many-to-one mapping from an infinite set to itself, the mapping ${G_F}:\mathcal{E }\to \mathcal{E}$ can be proved to be surjective. Because according to the definition of surjection, for every element  
\begin{multline*}
	(((s{{'}^{1}}s{{'}^2}\cdots ),(u{{'}^{1}}u{{'}^2}\cdots ),\ldots ,(v{{'}^{1}}v{{'}^2}\cdots )),\\(x{{'}_{1}},x{{'}_2}, \ldots ,x{{'}_m}))\in \mathcal{E}\,,
\end{multline*}
in the codomain of ${G_F}: \mathcal{E}\to \mathcal{E}$, there is at least one element in the domain of $G_F$, take
\begin{multline*}
(((0s{{'}^{1}}s{{'}^2}\cdots ),(0u{{'}^{1}}u{{'}^2}\cdots ),\ldots ,(0v{{'}^{1}}v{{'}^2}\cdots )),\\(x{{'}_{1}},x{{'}_2},\ldots ,x{{'}_m}))\in \mathcal{E}\,.
\end{multline*}
as an example, so the mapping ${G_F}:\mathcal{E}\to \mathcal{E}$ is surjective. Thus 
\begin{equation}
	\label{eq23}
G_F(\mathcal{E})=\mathcal{E}
\end{equation}  
is satisfied.    
According to Eq.\eqref{eq22}-\eqref{eq23},                                         
\begin{equation}
	\left\{
	\begin{IEEEeqnarraybox}[][c]{l}
	\IEEEstrut
	G_F^{n_0}(U)\! \!= \mathcal{E}\,, \\
	G_F^{n_0+1}(U) \!=\! G_F(G_F^{n_0}(U))\!=\!G_F(\mathcal{E})=\mathcal{E}\,, \\
	G_F^{n_0+2}(U)\!=\! G_F(G_F^{n_0+1}(U))\!=\!G_F(\mathcal{E})=\mathcal{E}\,, \\
			  \vdots \\
	G_F^n(U) \!=\! G_F(G_F^{n-1}(U))\!=\!\cdots\!=\!G_F(\mathcal{E})\!=\!\mathcal{E}(\forall n\geq{n_0})
	\IEEEstrut
	\end{IEEEeqnarraybox}
	\right.
	\label{eq24}
	\end{equation}
is obtained. Therefore, for any non-empty open sets $U,V\subset \mathcal{E}$, there is ${n_0}$, which always satisfies	
\begin{multline}
G_F^n(U)\cap V=\mathcal{E}\cap V=V\ne \varnothing (\forall n\geq{n_0})
\end{multline} 
This proves that ${G_F}$ is topologically mixing in the metric space $(\mathcal{E},d)$.                                
\end{proof}

\section{The loop-state contraction algorithm for constructing HDDCS}
\label{Construction of HDDCS}
In this section, the iterative function design is taken as the entry point. According to the selection method of the iterative function, the set of all possible iterative functions $S$ is obtained. Using the loop-state contraction algorithm, an iterative function $F\in S$ is constructed, so that the state transition diagram of $G_F$ is strongly connected, which ensures that the iterative equation $E^{k+1}=G_F( E^k)(k=0,1,2,\ldots)$ can meet the requirements of HDDCS.

\subsection{Selection method of iterative function $F$}
\label{The design of iterative function}
For a $m$-dimensional iterative function, we assume that $F_i(1\leqslant i\leqslant m)$ contains only $m$ items, and the $j(1\leqslant j\leqslant m)$-th item is an expression containing only one variable $x_j(1\leqslant j\leqslant m)$, and there are three possibilities for the expression of this variable $x_j(1\leqslant j\leqslant m)$. That is the original variable $x_j$, the inverse variable $\overline{x_j}$ or 0. The operators between them are bitwise AND ``$\cdot$", bitwise OR ``$+$" and bitwise exclusive OR ``$\oplus$". The three operators have the same precedence level and the order of operation is from left to right. According to the above selection method of iterative function, the set of all possible iterative functions $S$ is obtained.

For example, according to the above selection method of iterative function, a certain three-dimensional iterative function is obtained as
\begin{equation}
	\left\{
	\begin{IEEEeqnarraybox}[][c]{l}
	\IEEEstrut
	F_1(x_1, x_2, x_3)=  x_1  \oplus  x_2 + x_3 \,,\\
	F_2(x_1, x_2, x_3)= \overline{x_1} \cdot  x_3\,,\\
	F_3(x_1, x_2, x_3)=  \overline{x_1}  + \overline{x_2} \oplus  x_3\,.
	\IEEEstrut	
	\end{IEEEeqnarraybox} \right.
	\end{equation}
In the above function, the first item in $F_1(x_1, x_2, x_3)$ is $x_1$, the second item is $x_2$, and the third item is $x_3$. The operators between them are bitwise exclusive OR ``$\oplus$" and bitwise OR ``$+$". The first item in $F_2(x_1, x_2, x_3)$ is $\overline{x_1}$, the second item is 0, and the third item is $x_3$. The operator between them is bitwise AND ``$\cdot$". The first item in $F_3(x_1, x_2, x_3)$ is $\overline{x_1}$, the second item is $\overline{x_2}$, and the third item is $x_3$. The operators between them are the bitwise OR ``$+$" and the bitwise exclusive OR ``$\oplus$". The selection of the remaining iterative functions can be deduced by analogy.

\subsection{State transition table, state transition diagram and adjacency matrix}

For an iterative equation with precision $N$ and dimension $m$, the corresponding state number is $2^{mN}$. Given finite precision $N$ and dimension $m$, the corresponding state transition table, state transition diagram and adjacency matrix can be further obtained according to the iterative function.

For example, assume the finite precision $N=1(P=1,Q=0,P+Q=N)$ and the dimension $m=2$, first consider the iterative function  uncontrolled  by random  sequences~\cite{Wang:hddcs:IEEECSI2016}, as

\begin{equation}
\left\{
\begin{IEEEeqnarraybox}[][c]{l}
\IEEEstrut
	x_1^n=F_1(x_1^{n-1}, x_2^{n-1})=\overline{x_1^{n-1}}\oplus \overline{x_2^{n-1}}\,,\\
	x_2^n=F_2(x_1^{n-1}, x_2^{n-1})=\overline{x_1^{n-1}} + \overline{x_2^{n-1}}\,, 
\IEEEstrut	
\end{IEEEeqnarraybox} \right.
\label{2dexp}
\end{equation}	
According to Eq.\eqref{2dexp}, the corresponding state transition table, state transition diagram, and adjacency matrix can be obtained, as shown in Fig.~\ref{eq7state}. Note that the state transition table, state transition diagram, and adjacency matrix are different manifestations of the same thing. The other two forms can be deduced from one form, namely the adjacency matrix and the state transition diagram can be deduced from the state transition table, the state transition table and the state transition diagram can be deduced from the adjacency matrix, and the state transition diagram can be deduced separately from state transition table and adjacency matrix. For convenience, first obtain the state transition table according to the iterative equation, and then obtain the state transition diagram and adjacency matrix according to the state transition table.

\begin{figure}[!htb]
	\centering
	\includegraphics[width=0.6\figwidth]{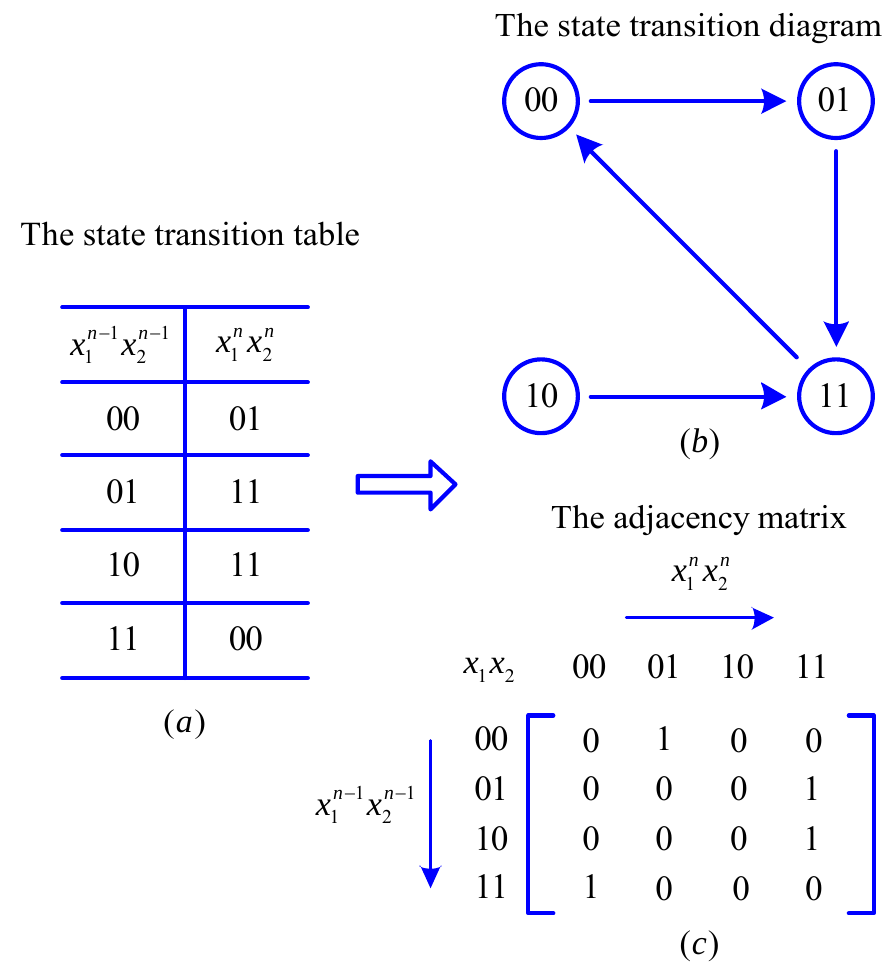}
	\caption{The corresponding state transition table, state transition diagram and adjacency matrix to Eq.~\eqref{2dexp}}
	\label{eq7state}
	\end{figure}
    
According to Eq.\eqref{2dexp}, the corresponding iterative function controlled by random sequences is
\begin{equation} \left\{
\begin{IEEEeqnarraybox}[][c]{l}
\IEEEstrut
x_1^n=x_1^{n-1}\cdot\overline{s^n}+((\overline{x_1^{n-1}}\oplus \overline{x_2^{n-1}})\cdot s^n)\,, \\
x_2^n=x_2^{n-1}\cdot\overline{u^n}+((\overline{x_1^{n-1}}+ \overline{x_2^{n-1}})\cdot u^n)\,, 
\IEEEstrut	
\end{IEEEeqnarraybox}\right.
\label{gfxy}
\end{equation}
where $s=s^1s^2s^3\cdots$ and $u=u^1u^2u^3\cdots$ are two random sequences. In the same way, the state transition table, state transition diagram, adjacency matrix, and standardized adjacency matrix corresponding to Eq.\eqref{gfxy} are obtained, as shown in Fig.~\ref{eq8state}.

\begin{figure}[!htb]
	\centering
	\includegraphics[width=\figwidth]{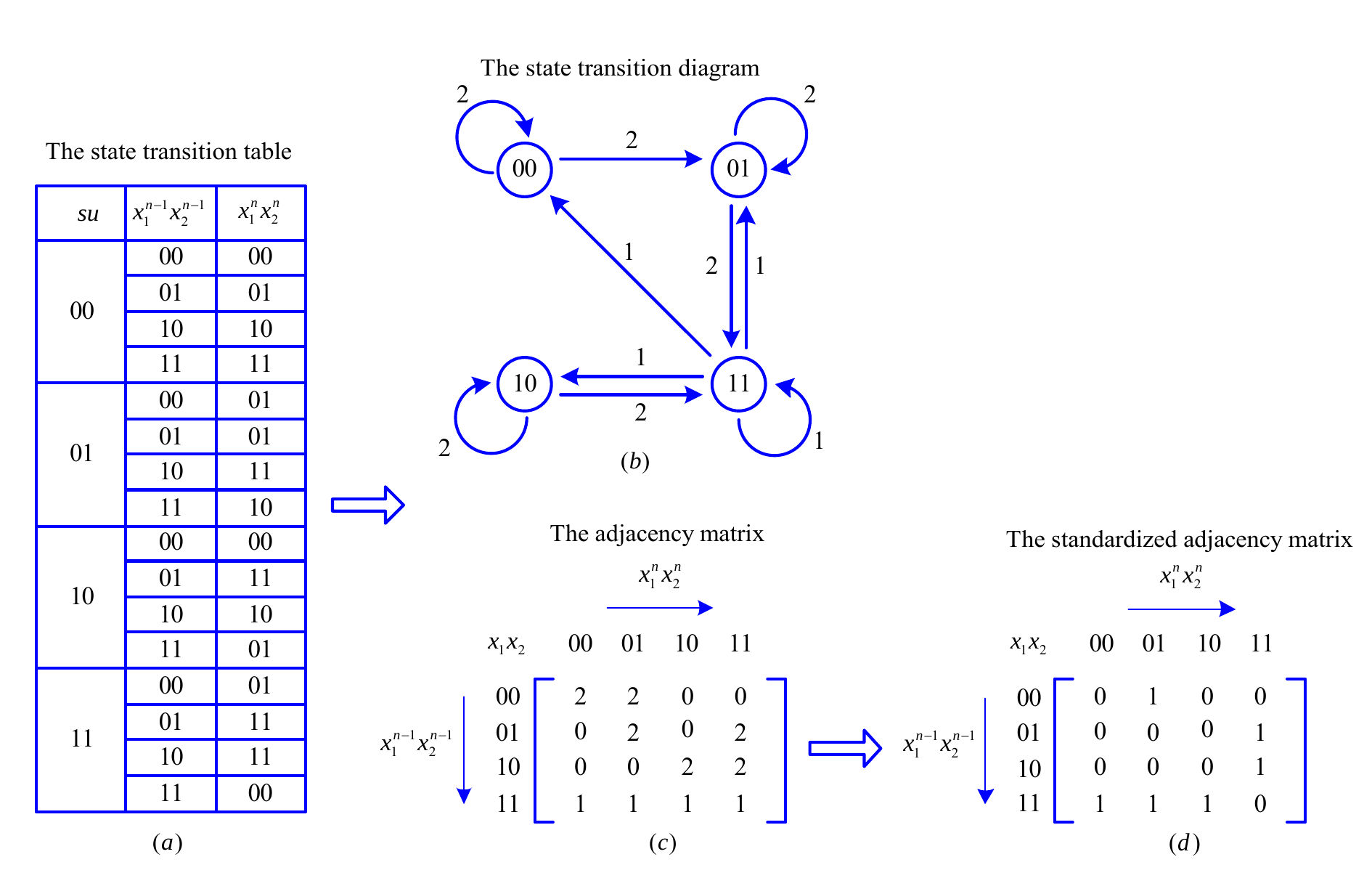}
	\caption{The corresponding state transition table, state transition diagram and adjacency matrix to Eq.~\eqref{gfxy}}
	\label{eq8state}
	\end{figure}

 It can be seen that the four numbers 2, 2, 2, and 1 on the main diagonal in the Fig.~\ref{eq8state}(c) respectively, which indicate that state 00 has 2 self-loops, state 01 has 2 self-loops, state 10 has 2 self-loops, and state 11 has 1 self-loop in Fig.~\ref{eq8state}(b). For the connectivity of the state transition diagram, the self-loops do not affect the strong connectivity, but they increase the complexity of the algorithm and may cause the program to fall into an endless loop. Therefore, all numbers on the main diagonal in the adjacency matrix should be set to 0. In addition, according to Fig.~\ref{eq8state}(b), there are two edges from the current state 00 to the next state 01, two edges from the current state 01 to the next state 11, and two edges from the current state 10 to the next state 11. For the connectivity of the state transition diagram, the situation that having multiple edges from the current state to the next state is equivalent to the situation that having only one edge from the current state to the next state, so all three numbers should be set to 1. Based on the above considerations, the corresponding standardized adjacency matrix is shown in Fig.~\ref{eq8state}(d). Hereinafter, the standardized adjacency matrix is abbreviated as the adjacency matrix.

\subsection{Three theorems on strong connectivity}
\begin{theorem}[\cite{Dharwadker:graphtheory2007}] 
	\label{thm3}
    If there is a loop in the state transition diagram that passes through each state at least once, it is strongly connected.
\end{theorem}

According to Theorem~\ref{thm3}, it can be seen that there is no loop through each state at least once at the state transition diagram shown in Fig.~\ref{eq7state}(b), so it is not strongly connected. But for the state transition diagram shown in Fig.~\ref{eq8state}(b), there is a loop that passes through each state at least once, so it is strongly connected.

Note that when the state transition diagram is relatively simple, we can directly use it to determine whether it has strong connectivity, but when it is more complex, we only use the state transition diagram to determine its strong connectivity, which is often more difficult. In order to further solve this problem, in the following discussion, we will use the adjacency matrix to judge the strong connectivity of the state transition diagram.

\begin{theorem} [\cite{Khuller:sccnon1998}] 
	\label{thm4}
    Given a state transition diagram, find a loop, contract the loop, and then recurse. Only the strongly connected diagram will eventually contract to a single state.
\end{theorem}
Theorem~\ref{thm4} was first proposed by Khuller et al. The goal of their algorithm proposed in \cite{Khuller:sccnon1998} is to find the minimum equivalent graph. In this article, we will judge the strong connectivity of the state transition diagram corresponding to the adjacency matrix according to the Theorem~\ref{thm4}.

\begin{theorem} 
	\label{thm5}
    When the finite precision $N=1$, if the operators between any two items of the iterative function $F$ are all bitwise operators, and the corresponding state transition diagram of ${G_F}$ is strongly connected, then when the finite precision $N>1$, the state transition diagram of ${G_F}$ corresponding to iterative function $F$ is still strongly connected.
\end{theorem}

\begin{proof}
Use ${F_1},{F_2},\ldots {F_m}$ to represent the $m$-dimensional iterative function~\cite{Wang:hddcs:IEEECSI2016}, and the binary form of each iteration value is 
\begin{equation}
	\left\{
	\begin{IEEEeqnarraybox}[][c]{ccccccccccc}
		\label{eq3.4}
	\IEEEstrut
	x_1 	& = 	& x_{1,P-1} & x_{1,P-2} & \cdots	& x_{1,0} 	&. 	& x_{1,-1} 	& x_{1,-2} 	&\cdots	& x_{1,-Q}\,,\\
	x_2 	& = 	& x_{2,P-1} & x_{2,P-2} & \cdots	& x_{2,0} 	&. 	& x_{2,-1} 	& x_{2,-2} 	&\cdots & x_{2,-Q}\,,\\
		&\vdots	&			&			&		 	&		&	&			&			&		&\\
	x_m 	& = 	& x_{m,P-1} & x_{m,P-2} & \cdots	& x_{m,0} 	&. 	& x_{m,-1} 	& x_{m,-2} 	&\cdots & x_{m,-Q}\,,
	\IEEEstrut
	\end{IEEEeqnarraybox} \right.
	\end{equation}
where $P+Q=N$.

According to Eq.\eqref{eq3.4}, set ${x_{1,i}}{x_{2,i}}\cdots {x_{m,i}}$ is composed of $${x_{1,i}},{x_{2,i}},\ldots,{x_{m,i} }(i=P-1,P-2,\ldots,0,-1,\ldots,-Q),$$ if the state transition diagram corresponding to one of the states $${x_{1,j}}{x_{2,j}}\cdots {x_{m,j}}(j\in \{P-1,P-2,\ldots ,0,-1,\ldots ,-Q\})$$ is strongly connected, then the corresponding state transition diagram for the rest of the state $${x_{1,k}}{x_{2,k}}\!\cdots\! {x_{m,k}}(k\!=\!P\!-\!1,P\!-\!2,\!\ldots\! ,0,\!-\!1,\!\ldots \!,-Q;k\!\ne\!j)$$ must also be strongly connected because of the bitwise operations.

Let's further prove that if the state transition diagram corresponding to one of the states $$x_{1,j}x_{2,j}\cdots x_{m,j}(j\in \{P-1,P-2,\ldots ,0,-1,\ldots ,-Q\})$$ is strongly connected, then the corresponding state transition diagram for $x_{1,P\!-\!1}\allowbreak x_{1,P\!-\!2}\cdots \allowbreak x_{1,-Q}x_{2,P\!-\!1}\allowbreak x_{2,P\!-\!2}\cdots \allowbreak x_{2,-Q}\cdots \allowbreak x_{m,P\!-\!1} \allowbreak x_{m,P\!-\!2}\cdots \allowbreak x_{m,-Q}$ is also strongly connected. The proof is as follows:
\begin{enumerate}
\item Let us consider the case of finite precision $N=1(P=1,Q=0)$ and suppose the state transition diagram corresponding to ${G_F}$ is strongly connected, that is, there is at least one path between any two states ${\hat{x}_{1,0}}{\hat{x}_{2,0}}\cdots {\hat{x}_{m,0}}$ and ${\tilde{x}_{1,0}}{\tilde{x}_{2,0}}\cdots {\tilde{x}_{m,0}}$.

\item For $N=2(P=2,Q=0)$, there is at least one path between any two the lowest bit states ${\hat{x}_{1,0}}{\hat{x}_{2,0}}\cdots {\hat{x}_{m,0}}$ and ${\tilde{x}_{1,0}}{\tilde{x}_{2,0}}\cdots {\tilde{x}_{m,0}}$,  ${k_0}$ is set to equal to the number of edges in the connected path between the state ${\hat{x}_{1,0}}{\hat{x}_{2,0}}\cdots {\hat{x}_{m,0}}$ and the state ${\tilde{x}_{1,0}}{\tilde{x}_{2,0}}\cdots {\tilde{x}_{m,0}}$. It means
\begin{equation}
	\label{Gfk5}
	\begin{IEEEeqnarraybox}[][c]{ll}
	\IEEEstrut
	&(G_F^{k_0}((\hat{s}, \hat{u}, \ldots, \hat{v}), \\
	&(\hat{x}_{1,1} \hat{x}_{1,0}, \hat{x}_{2,1}\hat{x}_{2,0},\ldots,\hat{x}_{m,1}\hat{x}_{m,0})))_{x_1,x_2,\ldots ,x_m}\\
	 = & (x'_{1,1} \tilde{x}_{1,0}, x'_{2,1}\tilde{x}_{2,0},\ldots,x'_{m,1}\tilde{x}_{m,0})\,.
\IEEEstrut
\end{IEEEeqnarraybox} 
\end{equation}
If one gets
\begin{equation*}
	\begin{IEEEeqnarraybox}[][c]{ll}
	\IEEEstrut
&\left( x{{'}_{1,1}}{\tilde{x}_{1,0}},x{{'}_{2,1}}{\tilde{x}_{2,0}},\ldots ,x{{'}_{m,1}}{\tilde{x}_{m,0}} \right)\\
=&\left( {\tilde{x}_{1,1}}{\tilde{x}_{1,0}},{\tilde{x}_{2,1}}{\tilde{x}_{2,0}},\ldots ,{\tilde{x}_{m,1}}{\tilde{x}_{m,0}} \right)\,,
\IEEEstrut
\end{IEEEeqnarraybox} 
\end{equation*}
it shows that any two states are connected by ${k_0}$ directed edges, and the state transition diagram of ${G_F}$ is strongly connected from Definition~\ref{strong connectivity}. If inequality
\begin{equation*}
	\begin{IEEEeqnarraybox}[][c]{ll}
	\IEEEstrut
	&\left( x{{'}_{1,1}}{\tilde{x}_{1,0}},x{{'}_{2,1}}{\tilde{x}_{2,0}},\ldots ,x{{'}_{m,1}}{\tilde{x}_{m,0}} \right)\\
	\neq&\left( {\tilde{x}_{1,1}}{\tilde{x}_{1,0}},{\tilde{x}_{2,1}}{\tilde{x}_{2,0}},\ldots ,{\tilde{x}_{m,1}}{\tilde{x}_{m,0}} \right)\,,
	\IEEEstrut
\end{IEEEeqnarraybox} 
\end{equation*}
exists, the connected path between between $ x{{'}_{1,1}}{\tilde{x}_{1,0}},x{{'}_{2,1}}{\tilde{x}_{2,0}},\ldots ,x{{'}_{m,1}}{\tilde{x}_{m,0}} $ and ${\tilde{x}_{1,1}}{\tilde{x}_{1,0}},{\tilde{x}_{2,1}}{\tilde{x}_{2,0}},\ldots ,{\tilde{x}_{m,1}}{\tilde{x}_{m,0}}$ should be considered. Note that ${G_F}$ is controlled by random sequences, so the lowest bits of all random numbers in $m$ random sequences can be set to 0, so that the lowest bits ${\tilde{x}_{1,0}},{\tilde{x}_{2,0}},\ldots ,{\tilde{x}_{m,0}}$ remain unchanged for the following iterations, and only the highest bits $x{{'}_ {1,1}},x{{'}_{2,1}},\ldots ,x{{'}_{m,1}}$ continue to participate in the update operation. And because the state formed by the highest bit is ${x_{1,1}}{x_{2,1}}\cdots {x_{m,1}}$, its corresponding state transition diagram is also strongly connected because of the bitwise operations, there is at least one path between $x{{'}_{1,1}}x{{'}_{2,1}}\cdots x{{'}_{m,1}}$ and ${\tilde{x}_ {1,1}}{\tilde{x}_{2,1}}\cdots {\tilde{x}_{m,1}}$, after another iteration of $i_0$ times,
where $i_0$ is equal to the number of edges in the connected path between the state $x{{'}_{1,1}}x{{'}_{2,1}}\cdots x{{'}_{m,1}}$ and $ {\tilde{x}_{1,1}}{\tilde{x}_{2,1}}\cdots {\tilde{x}_{m,1}}$. Equality
\begin{equation}
		\begin{IEEEeqnarraybox}[][c]{ll}
		\IEEEstrut
		&(G_F^{k_0+i_0}((\hat{s}, \hat{u}, \ldots, \hat{v}),\\ &(\hat{x}_{1,1} \hat{x}_{1,0}, \hat{x}_{2,1}\hat{x}_{2,0},\ldots,\hat{x}_{m,1}\hat{x}_{m,0})))_{x_1,x_2,\ldots ,x_m}
		 \\=&  (\tilde{x}_{1,1} \tilde{x}_{1,0},\tilde{x}_{2,1}\tilde{x}_{2,0},\ldots,\tilde{x}_{m,1}\tilde{x}_{m,0})
	\IEEEstrut
\end{IEEEeqnarraybox} 
\label{Gfk6}
\end{equation}
holds, which shows that any two states ${\hat{x}_{1,1}}{\hat{x}_{1,0}}{\hat{x}_{2,1}}{\hat{x}_{2,0}}\cdots {\hat{x}_{m,1}}{\hat{x}_{m,0}}$ and ${\tilde{x}_{1,1}}{\tilde{x}_{1,0}}{\tilde{x}_{2,1}}{\tilde{x}_{2,0}}\cdots {\tilde{x}_{m,1}}{\tilde{x}_{m,0}}$ are connected, and the state transition diagram of ${G_F}$ is strongly connected from Definition~\ref{strong connectivity}.

\item Consider the case of finite precision $N>2$, the current state gradually approaches the target state from the lowest bit (i.e. $-Q$ bit) to the highest bit (i.e. $P-1$ bit) through the control of some random sequences, and finally the path between any two states is found, thereby obtaining the conclusion that the state transition diagram of ${G_F}$ corresponding to function $F$ is still strongly connected.
\end{enumerate}	
\end{proof}
According to Theorem~\ref{thm5}, we do not need to use finite precision $N>1$ to analyze the strong connectivity of the state transition diagram of  ${G_F}$ corresponding to the iterative function $F$, and $N>1$ makes the corresponding number of states ${2^{mN}}\gg{2^m}$ and greatly increases the complexity of the algorithm. We only need to use finite precision $N=1$ to analyze the strong connectivity of the state transition diagram of  ${G_F}$ corresponding to the iterative function $F$, it makes the corresponding number of states ${2^{mN}}={2^m}$ the smallest, and also can greatly reduce the complexity of the algorithm.

\subsection{Loop-state  contraction algorithm}

According to Theorem \ref{thm3}-\ref{thm5}, several criteria for determining whether the state transition diagram is strongly connected and the contraction algorithm of the adjacency matrix are further given, and finally the flow chart of the loop-state contraction algorithm is obtained.

\begin{Judgment criteria}
The necessary condition for the  state transition diagram corresponding to the uncontracted adjacency matrix to be strongly connected is that there is at least one 1 in each row and each column. Otherwise, it is not strongly connected. Conversely, if there is at least one 1 in each row and each column of the adjacency matrix before contraction, it cannot guarantee that the corresponding state transition diagram is strongly connected.
\end{Judgment criteria}

The description of Judgment criteria 1 is as follows: it is necessary to note that there is a row of all 0s in the adjacency matrix, indicating that there is no directed edge from the state corresponding to this row in the state transition diagram. There is a column of all 0s in the adjacency matrix, indicating that there is no directed edge that reaches the state corresponding to this column in the state transition diagram.

For example, there is the third row of all 0s in the adjacency matrix as shown in Fig.~\ref{2ddcsstate3}(a), and the state corresponding to the third row is 10. So, it can be seen that there is no directed edge from state 10 in the state transition diagram in Fig.~\ref{2ddcsstate3}(b). Therefore, the  diagram is not strongly connected in Fig.~\ref{2ddcsstate3}(b). According to Fig.~\ref{2ddcsstate3}(c), there is the third column of all 0s in the adjacency matrix, and the state corresponding to the third column is 10, so it can be seen that there is no directed edges to state 10 in the state transition diagram in Fig.~\ref{2ddcsstate3}(d). Therefore, the state transition diagram shown in Fig.~\ref{2ddcsstate3}(d) is not strongly connected.

\begin{figure}[!htb]
	\centering
	\includegraphics[width=0.8\figwidth]{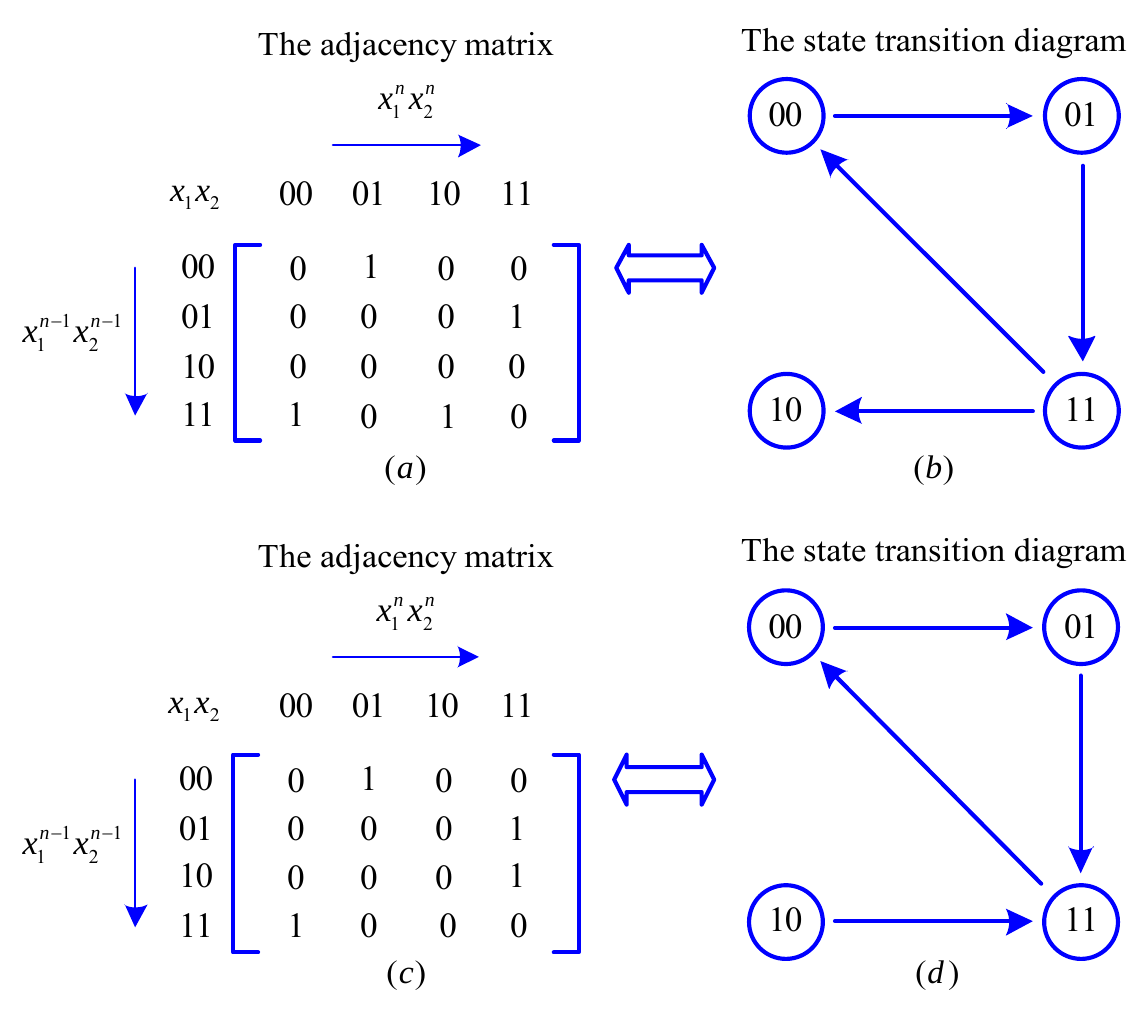}
	\caption{The adjacency matrix with the third row of all 0s and the third column of all 0s}
	\label{2ddcsstate3}
	\end{figure}
In addition,  the order of the uncontracted adjacency matrix satisfies $O={2^m}(m=2,3,4,\ldots )$ in Judgment criteria 1. Note that for an adjacency matrix with order ${2^m}(m=2,3,4,\ldots )$, even if there is at least one 1 in each row and each column of the adjacency matrix before contraction, there is no guarantee that the corresponding state transition diagram is strongly connected. 

\begin{Judgment criteria}
If the order of the contracted adjacency matrix satisfies $1<O<4$ and also satisfies that there is at least one 1 in each row and each column of the matrix, the corresponding state transition diagram is strongly connected. If the order of the contracted adjacency matrix satisfies $O=1$, the state transition diagram is strongly connected. Otherwise, it is not strongly connected.
\end{Judgment criteria}
In Judgment criterion 2, if the order of the contracted adjacency matrix  $O=1$, according to Theorem~\ref{thm4}, the state transition diagram is strongly connected. If the order of the contracted adjacency matrix  $O=2$, when it is satisfied that there is at least one 1 in each row and each column of the matrix, the corresponding state transition diagram must be strongly connected. If the order of the contracted adjacency matrix is $O=3$, when there is at least one 1 in each row and each column of the matrix, the corresponding state transition diagram is also strongly connected.

\begin{Contraction algorithm for the adjacency matrix} 
The adjacency matrix is contracted once, and the contracted block matrix $B$ is 
\begin{equation}
B=\left[ \begin{matrix}
   0 & P  \\
   Q & R  \\
\end{matrix} \right]                                          
\end{equation}
where the order of the block matrix $B$ is $O=k+1$, and $k$ represents the total number of the states not in the loop. The $P$ in the block matrix $B$ is a $1\times k$ matrix. After the contraction,  the values in the matrix $P$ are the column sums (logical or) of the elements whose row number is in the loop and the column number is not in the loop. The $Q$ in the block matrix $B$ is a $k\times 1$ matrix, and the values in $Q$ are the row sums  (logical or) of the elements whose column number is in the loop and the row number is not in the loop. The $R$ in the block matrix $B$ is a $k\times k$ matrix, and the values in $R$ are the original values of the elements whose row number and column number are not in the loop.
\end{Contraction algorithm for the adjacency matrix} 

According to Judgment criterion 1-2 and the contraction algorithm for the adjacency matrix, the flow chart of the loop-state contraction algorithm is shown in Fig.~\ref{loop7}.

	\begin{figure}[!htb]
		\centering
		\includegraphics[width=\figwidth]{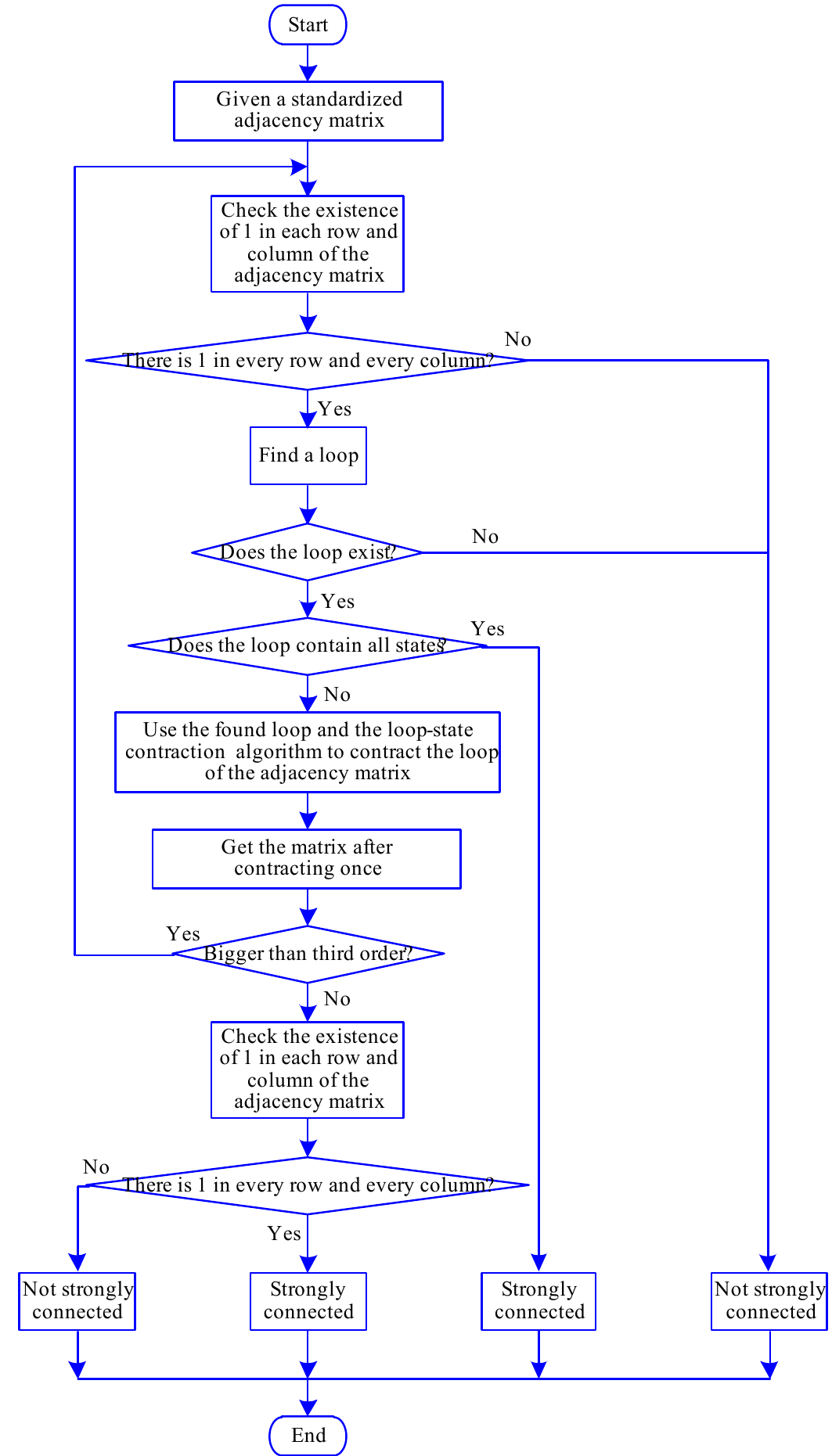}
		\caption{The flow chart of the loop-state contraction algorithm}
		\label{loop7}
	\end{figure}

In Fig.~\ref{loop7}, for a given uncontracted standardized adjacency matrix, first according to Judgment criterion 1, if there is a row or a column of all 0s in the adjacency matrix, the state transition diagram corresponding to the adjacency matrix is not strongly connected. If there is  at least one 1 in each row and each column, the adjacency matrix is contracted according to Contraction algorithm for the adjacency matrix until the order of the contracted adjacency matrix satisfies $O<4$. Then according to Judgment criterion 2, further check the strong connectivity of the adjacency matrix whose order satisfies $O<4$. In addition, in the Contraction algorithm for the adjacency matrix, a depth-first search (DFS) algorithm is used to find a loop. Specifically, the process of finding a loop is to start from a certain state and arrive at the next adjacent state that has not been visited in sequence, so as to explore as deeply as possible until a loop with a state number greater than two is found. If no loop is found with a state number greater than 2, then go back to the previous state, select another state that has not been visited as the starting point, and then repeat the above process until a loop with a state number greater than 2 is found. Fig.~\ref{loop8} shows the loop search process in the 8th order adjacency matrix. Start searching from state 000 and find the first unvisited adjacency state 001. The number of states in this loop is 2, so it should be discarded, and backtracking goes to state 000, find another unvisited adjacent state 100, and continue to search as deeply as possible. By analogy, the loop search process is obtained in Fig.~\ref{loop8}, where the dotted line indicates that the loop with the number of states equal to 2 is discarded.

\begin{figure}[!htb]
	\centering
	\includegraphics[width=\figwidth]{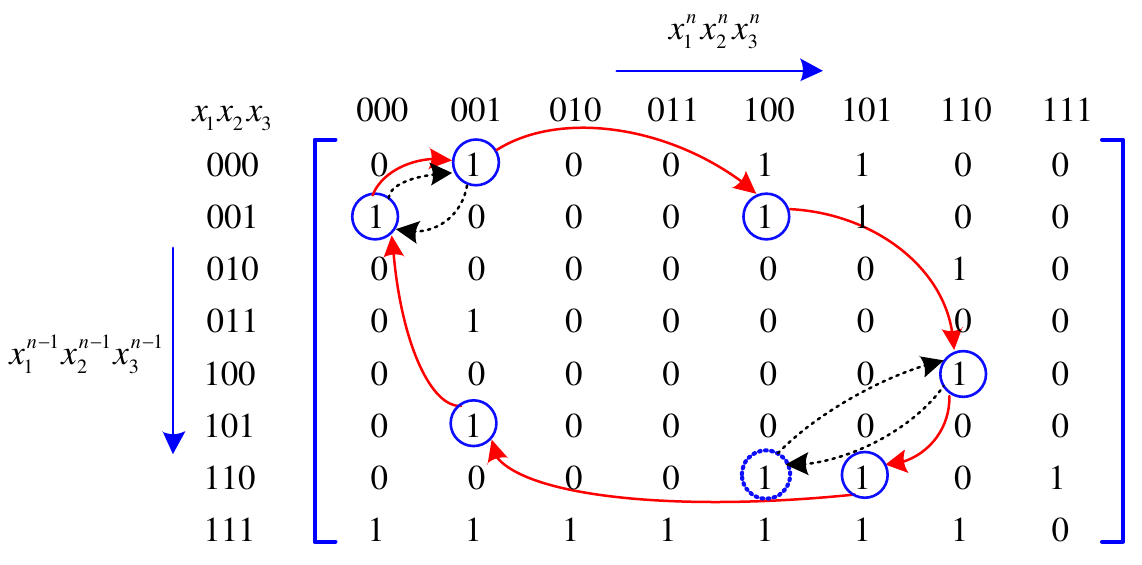}
	\caption{Depth-first search a loop for the adjacency matrix}
	\label{loop8}
\end{figure}

Here are some examples of using the flow chart of the loop-state contraction algorithm shown in Fig.~\ref{loop7} to analyze and judge whether the adjacency matrix is strongly connected.

\begin{example}
The adjacency matrix corresponding to Eq.~\eqref{2dexp} is shown in Fig.~\ref{eq7state}(c), and the flow chart of the loop-state contraction algorithm shown in Fig~\ref{loop7} is used for analysis. The third column in the adjacency matrix before contraction is all 0, so it is not strongly connected.
\end{example}
\begin{example}
The adjacency matrix corresponding to Eq.~\eqref{gfxy} is shown in Fig.~\ref{eq8state}(d), and the flow chart of the loop-state contraction algorithm shown in Fig.~\ref{loop7} is used for analysis, and then the process is shown in Fig.~\ref{loop9}-\ref{loop10}. The Fig.~\ref{loop9} indicates that a loop found as
$$(00,01)\to (01,11)\to (11,00)\to (00,01)$$
Fig.~\ref{loop10} shows that further using the found loop , according to the loop-state contraction algorithm, there is  at least one 1 in each row and each column of the 2nd order matrix after contraction. It can be seen that the state transition diagram of ${G_F}$ is strongly connected, which can ensure that Eq.~\eqref{gfxy} meets the requirements of HDDCS.

\begin{figure}[!htb]
	\centering
	\includegraphics[width=0.58\figwidth]{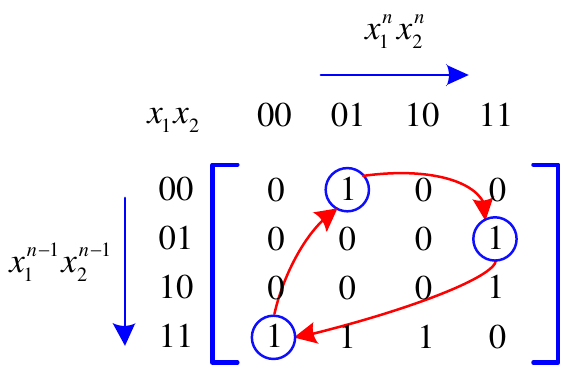}
	\caption{The adjacency matrix with one found loop}
	\label{loop9}
\end{figure}

\begin{figure}[!htb]
	\centering
	\includegraphics[width=0.9\figwidth]{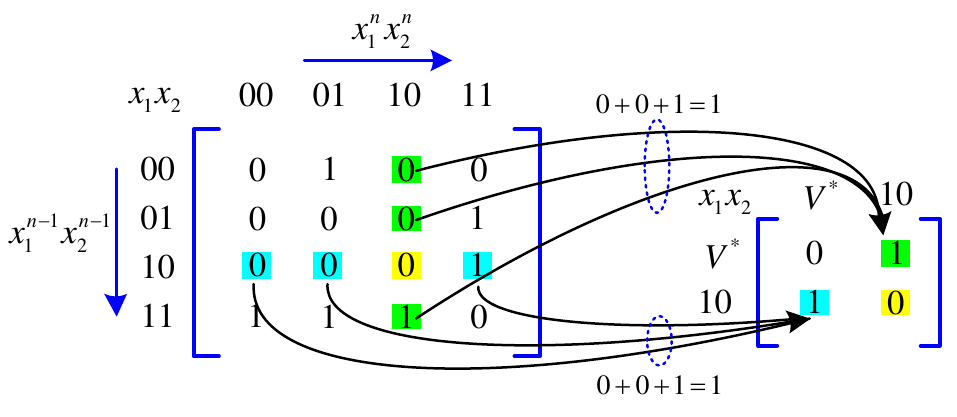}
	\caption{Contracting the 4th order adjacency matrix to the 2nd order adjacency matrix according to the loop-state contraction algorithm }
	\label{loop10}
\end{figure} 
\end{example}

\begin{example}
It is known that the iterative equation of a three-dimensional digital system that is not controlled by random sequences is
\begin{equation}
	\label{eq3-8}
	\left\{
	\begin{IEEEeqnarraybox}[][c]{l}
	\IEEEstrut
	x_1^n=F_1(x^{n-1}_1, x^{n-1}_2, x^{n-1}_3)=  \overline{x^{n-1}_1}  \cdot  \overline{x^{n-1}_2} + \overline{x^{n-1}_3} \,,\\
	x_2^n=F_2(x^{n-1}_1, x^{n-1}_2, x^{n-1}_3)=  \overline{x^{n-1}_1}  \oplus  \overline{x^{n-1}_2} \cdot \overline{x^{n-1}_3} \,,\\
	x_3^n=F_3(x^{n-1}_1, x^{n-1}_2, x^{n-1}_3)= \overline{x^{n-1}_1}  \oplus  \overline{x^{n-1}_2} \oplus \overline{x^{n-1}_3} \,,
	\IEEEstrut	
	\end{IEEEeqnarraybox} \right.
	\end{equation} 
According to Eq.~\eqref{eq3-8}, the corresponding iterative function controlled by random sequences is
\begin{equation}
	\label{eq3-9}
	\left\{
	\begin{IEEEeqnarraybox}[][c]{l}
	\IEEEstrut
	x_1^n= x_1^{n-1}\cdot \overline{s^n} +(\overline{x^{n-1}_1}  \cdot  \overline{x^{n-1}_2} + \overline{x^{n-1}_3})\cdot s^n \,,\\
	x_2^n= x_2^{n-1}\cdot \overline{u^n} +(\overline{x^{n-1}_1}  \oplus  \overline{x^{n-1}_2} \cdot \overline{x^{n-1}_3})\cdot u^n \,,\\
	x_3^n= x_3^{n-1}\cdot \overline{v^n}+(\overline{x^{n-1}_1}  \oplus  \overline{x^{n-1}_2} \oplus \overline{x^{n-1}_3})\cdot v^n \,,
	\IEEEstrut	
	\end{IEEEeqnarraybox} \right.
	\end{equation} 
 where $s={s^{1}}{s^2}{s^{3}}\cdots ,u={{u}^{1}}{{u}^2}{{u}^ {3}}\cdots ,v={{v}^{1}}{{v}^2}{{v}^{3}}\cdots $ are three random sequences. First, the standardized 8th order adjacency matrix corresponding to Eq.~\eqref{eq3-9}  is obtained as Fig.~\ref{loop11}(a). Secondly, according to the loop-state  contraction algorithm, the 8th order adjacency matrix is  contracted for the first time to obtain the 4th order adjacency matrix as shown in Fig.~\ref{loop11}(b). Finally, according to the loop-state  contraction algorithm, the 4th order adjacency matrix is further contracted at the second time into a 2nd order adjacency matrix, as shown in Fig.~\ref{loop11}(c). It can be seen that the state transition diagram of ${G_F}$ is strongly connected, which can ensure that Eq.~\eqref{eq3-9} meets the requirements of HDDCS.

\begin{figure}[!htb]
	\centering
	\includegraphics[width=\figwidth]{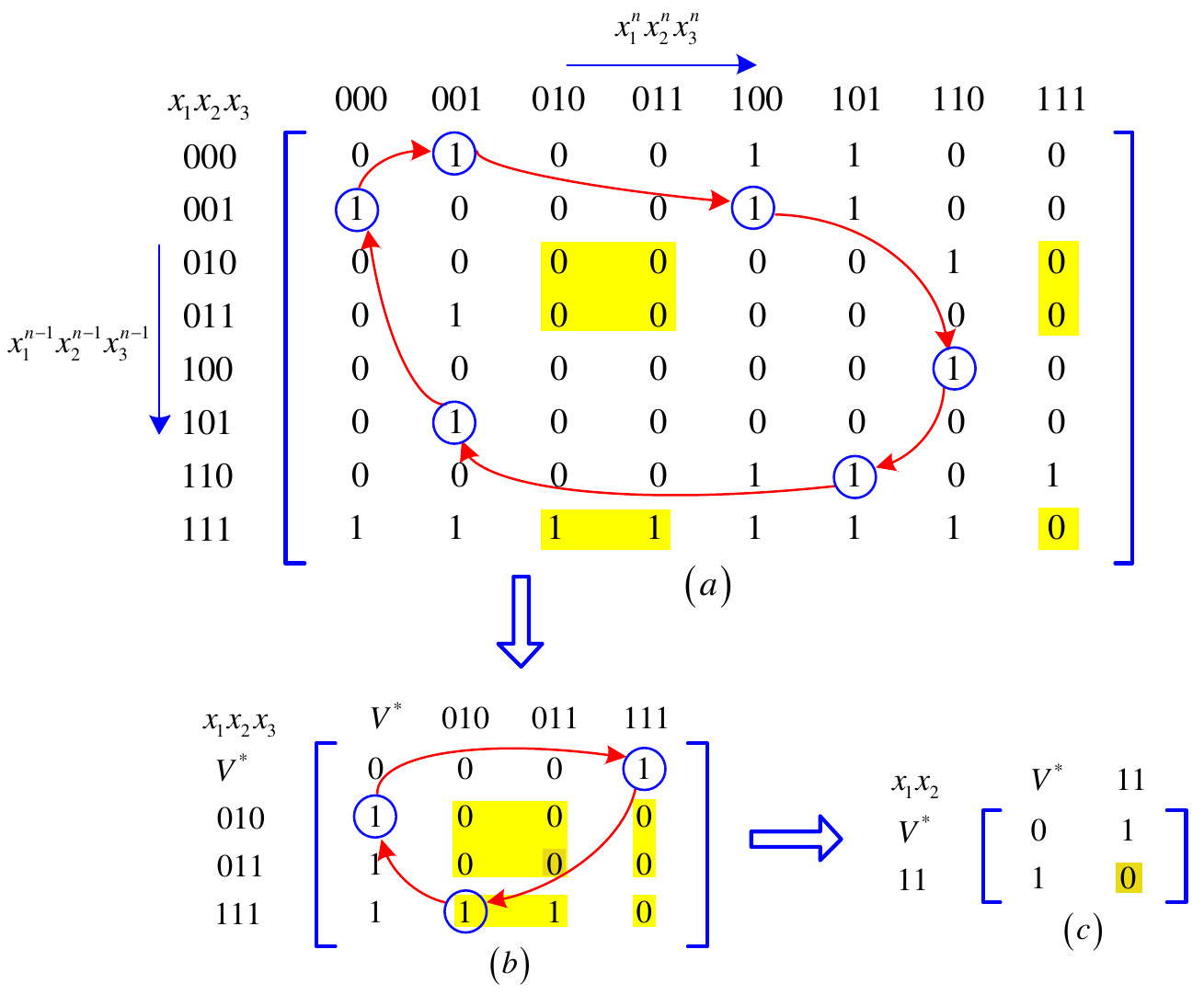}
	\caption{The 8th order adjacency matrix is contracted twice to a 2nd order adjacency matrix according to the loop-state  contraction algorithm}
	\label{loop11}
\end{figure}  
\end{example}

\subsection{Construct an iterative functions $F$ to meet the requirements of HDDCS by the loop-state  contraction algorithm}

In this section, we use the loop-state contraction algorithm to construct iterative functions $F$ that meets the requirements of HDDCS, which can ensure that the state transition diagram of ${G_F}$ is strongly connected.

According to the selection method of the iterative function in Section~\ref{The design of iterative function},  the exhaustive method is  firstly used to get all possible iterative functions $F\in S$ with the corresponding ${G_F}$. On this basis, all possible iterative equations controlled by random sequences ${E^{k+1}}={G_F}({E^{k}})(k=0,1,2, \ldots )$ can be further obtained. Secondly, according to the method shown in the Fig.~\ref{eq8state}, all possible adjacency matrices corresponding to ${E^{k+1}}={G_F}({E^{k}})(k=0,1,2,\ldots)$ are found.  Then, according to the loop-state contraction algorithm shown in Fig.~\ref{loop7}, all these possible adjacency matrices are analyzed to determine which has the  strongly connected state transition diagram of ${G_F}$. Finally, find out some iterative functions, which corresponding to  ${E^{k+1}}={G_F}({E^{k}} )(k=0,1,2,\ldots )$ can meet the requirements of HDDCS. Part of the 2-5 dimensional iterative functions $F$ constructed in this way are shown in Table~\ref{mddcs}.

\begin{table*}
	\centering
	\caption{Part of the 2-5 dimensional iterative functions $F$ constructed with the loop-state  contraction algorithm}
	\label{mddcs}
	\resizebox{\textwidth}{!}{
		\begin{tabular}{|l|l|l|l| }	\hline	
			$\langle 1 \rangle \left\{ {\begin{array}{*{20}{l}}
					{F_1(x_1, x_2) = x_1 \oplus \overline  x_2 }\\
					{F_2(x_1, x_2) =\overline x_1 \cdot \overline  x_2 }
			\end{array}} \right. $
			&$\langle 6 \rangle \left\{ {\begin{array}{*{20}{l}}
					{F_1(x_1, x_2) = \overline x_1  + \overline x_2 }\\
					{F_2(x_1, x_2) = x_1 \cdot \overline x_2 }
			\end{array}} \right. $ 
			& $\langle 11 \rangle \left\{ {\begin{array}{*{20}{l}}
					{F_1(x_1, x_2, x_3) = \overline x_1  \oplus \overline x_2 \cdot \overline x_3 }\\
					{F_2(x_1, x_2, x_3) = \overline x_1  \cdot \overline x_2 + \overline x_3}\\
					{F_3(x_1, x_2, x_3) = \overline x_1  \oplus \overline x_2 \oplus \overline x_3}
			\end{array}} \right. $	
			&$\langle 16 \rangle \left\{ {\begin{array}{*{20}{l}}
					{F_1(x_1, x_2, x_3, x_4) =  x_1 \oplus \overline x_2 \oplus  x_3 \cdot \overline x_4 }\\
					{F_2(x_1, x_2, x_3, x_4) =  x_1 \cdot \overline x_2 + x_3 \oplus x_4 }\\
					{F_3(x_1, x_2, x_3, x_4) =  x_1  + \overline x_2 + \overline x_3 \cdot \overline x_4}\\
					{F_4(x_1, x_2, x_3, x_4) =  x_1 \oplus \overline x_2 +  x_3 \oplus \overline x_4}
			\end{array}} \right. $	    \\ \hline
			
			$\langle 2 \rangle \left\{ {\begin{array}{*{20}{l}}
					{F_1(x_1, x_2) = \overline x_1  \oplus  x_2 }\\
					{F_2(x_1, x_2) =\overline x_1 \cdot \overline x_2}
			\end{array}} \right. $ 
			&$\langle 7 \rangle \left\{ {\begin{array}{*{20}{l}}
					{F_1(x_1, x_2) = x_1 \oplus \overline x_2 }\\
					{F_2(x_1, x_2) = \overline x_1  \cdot \overline x_2 }
			\end{array}} \right. $ 
			& $\langle 12 \rangle \left\{ {\begin{array}{*{20}{l}}
					{F_1(x_1, x_2, x_3) =  x_1  + \overline x_2 \cdot \overline x_3 }\\
					{F_2(x_1, x_2, x_3) =  x_1  \oplus \overline x_2 + \overline x_3}\\
					{F_3(x_1, x_2, x_3) = \overline x_1  \oplus \overline x_2 + \overline x_3}
			\end{array}} \right. $ 
			&$\langle 17 \rangle \left\{ {\begin{array}{*{20}{l}}
					{F_1(x_1, x_2, x_3, x_4) = \overline x_1 \cdot  x_2 \oplus \overline x_3 \cdot  x_4 }\\
					{F_2(x_1, x_2, x_3, x_4) =  x_1 \oplus \overline x_2 + \overline x_3 \oplus x_4 }\\
					{F_3(x_1, x_2, x_3, x_4) =  x_1  \cdot  x_2 \oplus  x_3 \cdot  x_4}\\
					{F_4(x_1, x_2, x_3, x_4) =  x_1 +  x_2 \cdot \overline x_3 \cdot \overline x_4}
			\end{array}} \right. $  \\ \hline		
			
			$\langle 3 \rangle \left\{ {\begin{array}{*{20}{l}}
					{F_1(x_1, x_2) = \overline x_1  \cdot  x_2 }\\
					{F_2(x_1, x_2) = \overline x_1 \oplus  x_2 }
			\end{array}} \right. $
			&$\langle 8 \rangle \left\{ {\begin{array}{*{20}{l}}
					{F_1(x_1, x_2) = \overline x_1  \cdot \overline x_2 }\\
					{F_2(x_1, x_2) = x_1 \cdot \overline x_2 }
			\end{array}} \right. $ 
			&$\langle 13 \rangle \left\{ {\begin{array}{*{20}{l}}
					{F_1(x_1, x_2, x_3, x_4) = \overline x_1 \cdot x_2 \oplus x_3 \cdot x_4 }\\
					{F_2(x_1, x_2, x_3, x_4) = \overline x_1 \cdot x_2 \oplus x_3 \oplus x_4 }\\
					{F_3(x_1, x_2, x_3, x_4) =  x_1  \oplus \overline x_2 + \overline x_3 \cdot  \overline x_4}\\
					{F_4(x_1, x_2, x_3, x_4) =  x_1 \cdot \overline x_2 \oplus  x_3 \oplus x_4}
			\end{array}} \right. $		
			&$\langle 18 \rangle \left\{ {\begin{array}{*{20}{l}}
					{F_1(x_1, x_2, x_3, x_4) = \overline x_1 \cdot \overline x_2 \oplus \overline x_3 +  x_4 }\\
					{F_2(x_1, x_2, x_3, x_4) = \overline x_1 \oplus \overline x_2 + x_3 \oplus x_4 }\\
					{F_3(x_1, x_2, x_3, x_4) =  x_1 \cdot \overline  x_2 \oplus  x_3 \cdot  x_4}\\
					{F_4(x_1, x_2, x_3, x_4) = \overline x_1 \cdot \overline x_2 + \overline x_3 \oplus  x_4}
			\end{array}} \right. $ \\\hline

			$\langle 4 \rangle \left\{ {\begin{array}{*{20}{l}}
					{F_1(x_1, x_2) = \overline x_1  \oplus \overline x_2 }\\
					{F_2(x_1, x_2) = \overline x_1  + \overline x_2 }
			\end{array}} \right. $	
			&$\langle 9 \rangle \left\{ {\begin{array}{*{20}{l}}
					{F_1(x_1, x_2) = \overline x_1  \cdot \overline x_2 }\\
					{F_2(x_1, x_2) = \overline x_1  \oplus x_2}
			\end{array}} \right. $     
			&$\langle 14 \rangle \left\{ {\begin{array}{*{20}{l}}
					{F_1(x_1, x_2, x_3, x_4) =  x_1 \cdot x_2 \oplus x_3 + x_4 }\\
					{F_2(x_1, x_2, x_3, x_4) = \overline x_1 \oplus x_2 \cdot x_3 \oplus x_4 }\\
					{F_3(x_1, x_2, x_3, x_4) =  x_1  \cdot \overline x_2 \oplus \overline x_3 \cdot  x_4}\\
					{F_4(x_1, x_2, x_3, x_4) = \overline x_1 + \overline x_2 + \overline x_3 \oplus x_4}
			\end{array}} \right. $ 
			&$\langle 19 \rangle \left\{ {\begin{array}{*{20}{l}}
					{F_1(x_1, x_2, x_3, x_4, x_5) =  x_1 +  x_2  + x_3 + \overline x_4 \oplus x_5}\\
					{F_2(x_1, x_2, x_3, x_4, x_5) =  x_1 +  x_2 \oplus \overline x_3 \oplus x_4 \cdot \overline x_5 }\\
					{F_3(x_1, x_2, x_3, x_4, x_5) =  x_1 + \overline  x_2 + \overline  x_3 + \overline x_4 \oplus x_5}\\
					{F_4(x_1, x_2, x_3, x_4, x_5) =  x_1 + \overline x_2 +  x_3 \oplus  x_4 + \overline x_5}\\
					{F_5(x_1, x_2, x_3, x_4, x_5) = \overline x_1 + \overline x_2 + \overline x_3 + \overline x_4 \oplus x_5}
			\end{array}} \right. $  \\ \hline	
			
			$\langle 5 \rangle \left\{ {\begin{array}{*{20}{l}}
					{F_1(x_1, x_2) = \overline x_1  \oplus \overline x_2 }\\
					{F_2(x_1, x_2) = x_1 + \overline x_2 }
			\end{array}} \right. $ 
			&$\langle 10 \rangle \left\{ {\begin{array}{*{20}{l}}
					{F_1(x_1, x_2, x_3) =  x_1  \oplus  x_2 + x_3 }\\
					{F_2(x_1, x_2, x_3) = \overline x_1 \cdot  x_3}\\
					{F_3(x_1, x_2, x_3) =  \overline x_1  + \overline x_2 \oplus  x_3}
			\end{array}} \right. $	    
			&$\langle 15 \rangle \left\{ {\begin{array}{*{20}{l}}
					{F_1(x_1, x_2, x_3, x_4) =  x_1 \oplus x_2 \oplus x_3 + x_4 }\\
					{F_2(x_1, x_2, x_3, x_4) = \overline x_1 \oplus x_2 + x_3 \oplus x_4 }\\
					{F_3(x_1, x_2, x_3, x_4) =  x_1  + \overline x_2 \oplus \overline x_3 \cdot  x_4}\\
					{F_4(x_1, x_2, x_3, x_4) =  x_1 + \overline x_2 +  x_3 \oplus x_4}
			\end{array}} \right. $	   
			&$\langle 20 \rangle \left\{ {\begin{array}{*{20}{l}}
					{F_1(x_1, x_2, x_3, x_4, x_5) = \overline x_1 + \overline x_2  + x_3 \oplus \overline x_4 \oplus x_5}\\
					{F_2(x_1, x_2, x_3, x_4, x_5) = \overline x_1 + \overline x_2 \oplus \overline x_3 \oplus x_4 \oplus \overline x_5 }\\
					{F_3(x_1, x_2, x_3, x_4, x_5) = \overline x_1 \oplus \overline  x_2 + \overline  x_3 \cdot \overline x_4 \oplus x_5}\\
					{F_4(x_1, x_2, x_3, x_4, x_5) =  x_1 + \overline x_2 \oplus  x_3 \oplus  x_4 + \overline x_5}\\
					{F_51(x_1, x_2, x_3, x_4, x_5) = \overline x_1 + \overline x_2 + \overline x_3 + \overline x_4 \oplus x_5}
			\end{array}} \right. $		    \\ \hline
			
		\end{tabular}
	}
\end{table*}

\section{Chaotic ESN and its application in Mackey-Glass time series prediction}
\label{Prediction}
This section uses the adjacency matrix corresponding to HDDCS to construct the chaotic ESN to predict the Mackey-Glass time series. The experimental results show that the chaotic ESN constructed by the higher-dimensional system has better predictive performance than the chaotic ESN constructed by the lower-dimensional system when the size of reservoir is fixed. 

\subsection{Chaotic ESN and its construction method}
Traditional ESN uses a  reservoir composed of randomly sparsely connected neurons as the hidden layer. The input and feedback weights are initialized to random values, and the spectral radius of the  reservoir is pre-defined to guarantee the stability of the network. During network training, only the connection weights from the hidden layer to the output layer need to be trained. Assuming that the number of neurons in the input layer, reserve pool, and output layer of the ESN network are 1, M, and 1, respectively. The state update equation of the traditional ESN is

\begin{equation}
\label{eq4-1}
Z(t)=\tanh (WZ(t-1)+{W^{in}}I(t)+{W^{fb}}O(t-1)+V(t))\,,
\end{equation} 
where $t=1,2,3,\ldots$, $Z(t)={(Z_{1}(t),Z_2(t),\ldots,Z_M(t))}^T$ is the $M$-dimensional reservoir state vector at time $t$ with ${Z_i}(t)\in (-1,+1),i=1,2,\ldots ,M$, $Z(t-1)$ is the $M$-dimensional reservoir state vector at time $t-1$, $Z(0)$ is the pre-set initial value, $\tanh (\cdot )$ is a vector-valued nonlinear activation function, $W$ is the $M\times M$ reservoir weight matrix with a spectral radius smaller than unity, $W^{in}$ is the $M\times 1$ input weight matrix, ${W^{fb}}$ is $M$-dimensional weight vector for feedback connections from the output neuron to the reservoir, the value of each element in ${W^{in}}$ and ${W^{fb}}$ obeys a uniform distribution in the interval $[-1,+1)$, note that the matrix $W,{W^ {fb}},{W^{in}}$ are determined before training and remain unchanged during training and prediction, $I(t)$ is the $1$-dimensional input signal at time $t$, $O(t)$ is the $1$-dimensional output signal at time $t$, $O(0)$ is the pre-set initial value, and $V(t)$ is $M$-dimensional noise vector uniformly distributed in $[-1.0\times10^{-10},+1.0 \times 10^{-10})$.

In Eq.\eqref{eq4-1}, the output $O(t)$ is

\begin{equation}
	\label{eq4-2}
	O(t)=\tanh ({W^{out}}(Z(t),I(t)))\,,
\end{equation}
where $I(t)$ and $O(t)$ are both 1-dimensional signals at time $t$, $(Z(t),I(t))={{({Z_{1}}(t ),{Z_2}(t),\ldots ,{Z_M}(t),I(t))}^{T}}$ is the $M$-dimensional reservoir state vector, ${W^{out}}$ is $(M+1)$-dimensional weight vector for connections from the reservoir to the output neuron. Note that ${W^{out}}$ is determined after the training.

In Eq.~\eqref{eq4-1}, the reservoir weight matrix $W$ is
\begin{equation}
	\label{eq4-3}
W=k({W_{R}}\odot {W_{S}}/|{{\lambda }_{\max }}|)\,,
\end{equation}
where $\odot $ means that the elements of the same row and column in two matrices ${W_{R}}$ and ${W_{S}}$ are multiplied, ${W_{R}}$ is a random matrix with the values uniformly distributed in intervalv $[-0.5,+0.5)$, ${W_{S}}$ is a matrix with randomly generated elements with values of 0 or 1, ${{\lambda }_{\max }}$ is the largest eigenvalue of the matrix ${W_{R}}\odot {W_{S}}$, $k$ is the spectral radius of $W$, $0<k<1$ is set to guarantee the reservoir to work in a stable region.

It is well-known that the richer the dynamic behaviors of the ESN, the better the  performance of corresponding network is, and the dynamic behaviors of the network are related to the network structure of the reservoir.
Note that in the traditional ESN represented by Eq.\eqref{eq4-1}-\eqref{eq4-3}, the reservoir weight matrix $W$ is confirmed by calculation ${W_{R}}\odot {W_{S}}$, because the matrix ${W_{R}}$ and ${W_{S}}$ are random matrices, so $W$ usually does not guarantee the strong connectivity.
In order to obtain better performance, according to Eq.~\eqref{eq4-3}, we use the adjacency matrix ${{A}_{M\times M}} $ corresponding to ${G_F}$ in HDDCS designed in Section~\ref{Construction of HDDCS}  to replace the ${W_{S}}$ in Eq.~\eqref{eq4-3},
where ${{A}_{M\times M}}$ is a square matrix with order $M={2^{mN}}$, $m$ is the dimension of ${G_F}$, and $N $ is the finite precision.
For given $m$ and $N$, according to the iterative function given in the table \ref{mddcs}, the corresponding HDDCS and  the corresponding adjacency matrix ${{A}_ {M\times M}}$ can be obtained, and according to Theorem~\ref{thm5}, it can be ensured that the corresponding adjacency matrix ${{A}_{M\times M}}$ under arbitrary $N$ is still strongly connected.
The literature ~\cite{Bahi:Neural:2012} has proved that the ESN constructed by the adjacency matrix ${{A}_{M\times M}}$ with strong connectivity  satisfies the chaotic definition of Devaney.
Therefore, when we use the adjacency matrix ${{A}_{M\times M}}$ to replace the ${W_{S}}$ in Eq.\eqref{eq4-3}, the traditional ESN expressed by  Eq.\eqref{eq4-1}-\eqref{eq4-2} is transformed into a chaotic ESN.

\subsection{Mackey-Glass time series prediction}
Time series data often have the characteristics of high noise, randomness and nonlinearity. Its modeling, analysis and prediction problems have always been the hotspots of academic research. Typically, in order to predict time series more accurately, time series models are required to have both good nonlinear approximation capabilities and good memory capabilities. In order to solve this problem, artificial intelligence methods such as support vector networks and neural networks have been introduced into the field of time series analysis. The Mackey-Glass time series prediction has now become a typical benchmark problem for verifying the processing capabilities of neural networks~\cite{Shi:MGS:Ieee2007}.

The Mackey–Glass time series is deduced from a time-delay differential system with the form~\cite{Ren:Performance:ITC2020}
\begin{equation}
	\label{eq4-4}
\frac{\mathrm{d} P(t)}{\mathrm{d} t}=\frac{\beta P(t-\alpha )}{1+P{{(t-\alpha )}^{10}}}-\gamma P(t)\,,
\end{equation}
where $\alpha $, $\beta $, $\gamma $ are real numbers.

According to \cite{Jaeger:ESN:2004}, the values standardly employed in most of the Mackey–Glass time series prediction literature are $\alpha =17$, $\beta =0.2$, and $\gamma =0.1$. Then through Eq.\eqref{eq4-4}, we can get the time series data $P(1),P(2),\ldots ,P(3000)$ containing 3000 time points. Among them, noise $V(t)$ is added to the data of the first 2000 time points, which is used to train the chaotic ESN to determine the value of ${W^{out}}$. However, no noise $V(t)$ is added to the data at the following 1000 time points, which is used for Mackey-Glass time series prediction.

According to Eq.~\eqref{eq4-1}-\eqref{eq4-4}, in the training stage, the number of training time points is usually selected to be 2000. First set $Z(0)=0$, $I(1)=I(2)=\cdots=I(2000)=0.02$, and $O(0)=0, O(1)=P(1), O(2)=P(2), \ldots,O(2000)=P(2000)$, then iteratively get $Z(1), Z(2) ,\ldots, Z(2000)$ by Eq.~\eqref{eq4-1}. Let $U$ be a matrix of order $(M+1)\times 2000$, namely
	\begin{equation*}
		        U= \begin{pmatrix}		
			Z_1(1)	& Z_1(2) 	& \ldots 	& Z_1(2000) \\
			Z_2(1) 	& Z_2(2) 	& \ldots 	& Z_2(2000) \\
			\vdots 	& \vdots 	& \ddots 	& \vdots \\
			Z_M(1) 	& Z_M(2) 	& \ldots 	& Z_M(2000)\\
			I(1)   	& I(2)		& \ldots	& ~~I(2000)	
		        \end{pmatrix}\,.
	\end{equation*}
Then according to Eq.\eqref{eq4-2}, we have
\begin{equation}
	\label{eq4-5}
	\begin{IEEEeqnarraybox}[][c]{ll}
	\IEEEstrut
&({\tanh}^{-1}(O(1)),{\tanh}^{-1}(O(2)),\ldots,{\tanh}^{-1}(O(2000)))\\
	=&W^{out}U\,,
	\IEEEstrut
\end{IEEEeqnarraybox} 
\end{equation}
Finally, using the pseudo-inverse operation of the matrix, the output $1\times (M+1)$ weight matrix ${W^{out}}$ is
\begin{equation}
	\label{eq4-6}
	\begin{IEEEeqnarraybox}[][c]{ll}
	\IEEEstrut
	&W^{out}\\
	=&(\!{\tanh}^{-1}(\!O(1)\!)\!,{\tanh}^{-1}(\!O(2)\!)\!,\ldots,{\tanh}^{-1}(\!O(2000)\!)\!)U^+\,,
	\IEEEstrut
\end{IEEEeqnarraybox} 
\end{equation}
where ${U^{+}}$ represents the pseudo-inverse operation of the  $(M+1)\times 2000$ matrix $U$. After the operation, the order of the matrix is $2000\times (M+1)$.

According to Eq.~\eqref{eq4-1}-\eqref{eq4-2}, in the prediction stage, first set $I(2001)=I(2002)=\cdots=I(3000)=0.02$, $O(2000)=P(2000)$ and $V(2001)=V(2002)=\cdots=V(3000)=0$, and substitute the $Z(2000)$ obtained from the training stage into Eq.~\eqref{eq4-1} to get $Z(2001)$, and then substituting $Z(2001)$ into Eq.~\eqref{eq4-2} to get $O(2001)$. Substitute the obtained $Z(2001)$ and $O(2001)$ into Eq.~\eqref{eq4-1}-\eqref{eq4-2} to get the corresponding outputs $Z(2002)$ and $O(2002)$, the rest can be deduced by analogy. Through this iterative method, the available Mackey-Glass time prediction sequence is $O(2001),O(2002),\ldots,O(3000)$.

\renewcommand\arraystretch{1.6}
\begin{table*}[!htb]
	\centering
	\caption{comparisons of prediction performance for makey-glass time series using chaotic ESN with different iterative equations}
	\label{ddfc}
	\setlength\tabcolsep{1pt}
	\scalebox{0.8}{
	\begin{tabular}{|l|c|c|c|c|}	\hhline{|-----|}	
		\diagbox[dir=SE]{HDDCS and the adjacency matrix}{RMSE}{Size of the reservoir $M$} & $2^{8}$  & $2^{9}$ & $2^{10} $ & $2^{12}$ \\ \hhline{|=====|}  
		
		$ \left\{ {\begin{array}{*{20}{l}}
				{{x_{1}^n}= x_1^{n-1}\cdot\overline{s^{n}}+(( {x_{1}^{n-1}} \oplus \overline {x_{2}^{n-1}})\cdot s^n)}\\
				{{x_{2}^n}=x_2^{n-1}\cdot\overline{u^{n}}+ ((\overline{x_{1}^{n-1}}\cdot \overline  {x_{2}^{n-1}})\cdot u^n)}
		\end{array}} \right.\Rightarrow A_{M\times M}^{\langle 1 \rangle } $
		& $4.76 \times 10^{-3}$ & None & $4.81\times 10^{-4}$ & $4.06\times 10^{-4}$		\\ \hhline{|-----|}
		$ \left\{ {\begin{array}{*{20}{l}}
				{{x_{1}^n}= x_1^{n-1}\cdot\overline{s^{n}}+((\overline {x_{1}^{n-1}} \oplus  {x_{2}^{n-1}})\cdot s^n)}\\
				{{x_{2}^n}= x_2^{n-1}\cdot\overline{u^{n}}+((\overline {x_{1}^{n-1}} \cdot \overline {x_{2}^{n-1}})\cdot u^n)}
		\end{array}} \right. \Rightarrow A_{M\times M}^{\langle 2 \rangle }$
		&$1.61\times 10^{-2}$ & None & $1.23\times 10^{-3}$ & $3.45 \times 10^{-4}$		\\ \hhline{|-----|}
		$ \left\{ {\begin{array}{*{20}{l}}
				{{x_{1}^n}= x_1^{n-1}\cdot\overline{s^{n}}+((\overline {x_{1}^{n-1}}\cdot  {x_{2}^{n-1}})\cdot s^n)}\\
				{{x_{2}^n}= x_2^{n-1}\cdot\overline{u^{n}}+((\overline {x_{1}^{n-1}}\oplus  {x_{2}^{n-1}})\cdot u^n)}
		\end{array}} \right. \Rightarrow A_{M\times M}^{\langle 3 \rangle }$
		&$4.34 \times 10^{-2}$ & None & $2.27 \times 10^{-4}$ & $2.87\times 10^{-4}$	    \\ \hhline{|=====|}

		$ \left\{ {\begin{array}{*{20}{l}}
				{{x_{1}^n}=x_1^{n-1}\cdot\overline{s^{n}}+ ((\overline {x_{1}^{n-1}} \oplus \overline {x_{2}^{n-1}} \cdot \overline {x_{3}^{n-1}})\cdot s^n)}\\
				{{x_{2}^n}=x_2^{n-1}\cdot\overline{u^{n}}+ ((\overline {x_{1}^{n-1}} \cdot \overline {x_{2}^{n-1}} + \overline  {x_{3}^{n-1}})\cdot u^n)}\\
				{{x_{3}^n}=x_3^{n-1}\cdot\overline{v^{n}}+ ((\overline { x_{1}^{n-1}} \oplus \overline {x_{2}^{n-1}} \oplus \overline {x_{3}^{n-1}})\cdot v^n)}
		\end{array}} \right. \Rightarrow A_{M\times M}^{\langle 11 \rangle }$
		& None & $5.71\times 10^{-4}$ & None & $1.28\times 10^{-4}$      \\\hhline{|-----|} 
		$ \left\{ {\begin{array}{*{20}{l}}
				{{x_{1}^n}=x_1^{n-1}\cdot\overline{s^{n}}+ (( {x_{1}^{n-1}} + \overline {x_{2}^{n-1}}\cdot \overline {x_{3}^{n-1}})\cdot s^n)}\\
				{{x_{2}^n}=x_2^{n-1}\cdot\overline{u^{n}}+ (( {x_{1}^{n-1}} \oplus \overline {x_{2}^{n-1}}+ \overline {x_{3}^{n-1}})\cdot u^n)}\\
				{{x_{3}^n}=x_3^{n-1}\cdot\overline{v^{n}}+ ((\overline {x_{1}^{n-1}} \oplus \overline {x_{2}^{n-1}}+ \overline {x_{3}^{n-1}})\cdot v^n)}
		\end{array}} \right. \Rightarrow A_{M\times M}^{\langle 12 \rangle }$	
		& None & $1.10 \times 10^{-3}$ & None & $2.75 \times 10^{-4}$      \\ \hhline{|=====|}

		$ \left\{ {\begin{array}{*{20}{l}}
				{{x_{1}^n}=x_1^{n-1}\cdot\overline{s^{n}}+ ((\overline {x_{1}^{n-1}}\cdot {x_{2}^{n-1}}\oplus {x_{3}^{n-1}}\cdot {x_{4}^{n-1}})\cdot s^n)}\\
				{{x_{2}^n}=x_2^{n-1}\cdot\overline{u^{n}}+ ((\overline {x_{1}^{n-1}}\cdot {x_{2}^{n-1}}\oplus {x_{3}^{n-1}}\oplus {x_{4}^{n-1}})\cdot u^n)}\\
				{{x_{3}^n}=x_3^{n-1}\cdot\overline{v^{n}}+  (({x_{1}^{n-1}} \oplus \overline {x_{2}^{n-1}}+ \overline {x_{3}^{n-1}}\cdot \overline {x_{4}^{n-1}})\cdot v^n)}\\
				{{x_{4}^n}=x_4^{n-1}\cdot\overline{w^{n}}+  (({x_{1}^{n-1}}\cdot \overline {x_{2}^{n-1}}\oplus  {x_{3}^{n-1}}\oplus {x_{4}^{n-1}})\cdot w^n)}
		\end{array}} \right. \Rightarrow A_{M\times M}^{\langle 13 \rangle }$	
		
		& $1.90 \times 10^{-3}$ & None &  None  & $1.08\times 10^{-4}$  \\ \hhline{|-----|}	
	\end{tabular}	}
\end{table*}

Fig.~\ref{ESN for MGS prediction}(a)-(b) respectively show the prediction results obtained by traditional ESN and chaotic ESN after 2000 training time points. The iterative function $\langle 13 \rangle$ from Table~\ref{mddcs}  is used to obtain the corresponding HDDCS, then the adjacency matrix ${{A}_{M\times M}}$ corresponding to the above HDDCS is used to construct the network structure of the reservoir, which makes ESN be a chaotic ESN.

As $m=4$ and $N=3$, the order of the adjacency matrix ${{A}_{M\times M}}$ corresponding to the iterative function $\langle 13 \rangle$ from Table~\ref{mddcs} is $M\times M=4096\times 4096$. In Fig.~\ref{ESN for MGS prediction}(a),  the black solid line is correct continuation of the Mackey-Glass time series, the blue dotted line is the chaotic ESN prediction of the Mackey-Glass time series, and the red dashed line is the traditional ESN prediction of the Mackey-Glass time series.
It should be noted that the 1000 time points of both the traditional ESN prediction and the chaotic ESN prediction all  fit the correct continuation of the Mackey-Glass time series well.
Fig.~\ref{ESN for MGS prediction}(b) illustrates the prediction error $O(t)-P(t)$ development on long prediction runs for the traditional ESN and chaotic ESN, both of them can achieve high prediction accuracies within 500 time points, but the long-term prediction performance of chaotic ESN is relatively better.

\begin{figure}[!htb]
	\centering
	\subfigure[Prediction of the Mackey-Glass time series for the traditional ESN and chaotic ESN]{
	  \includegraphics[width=8cm]{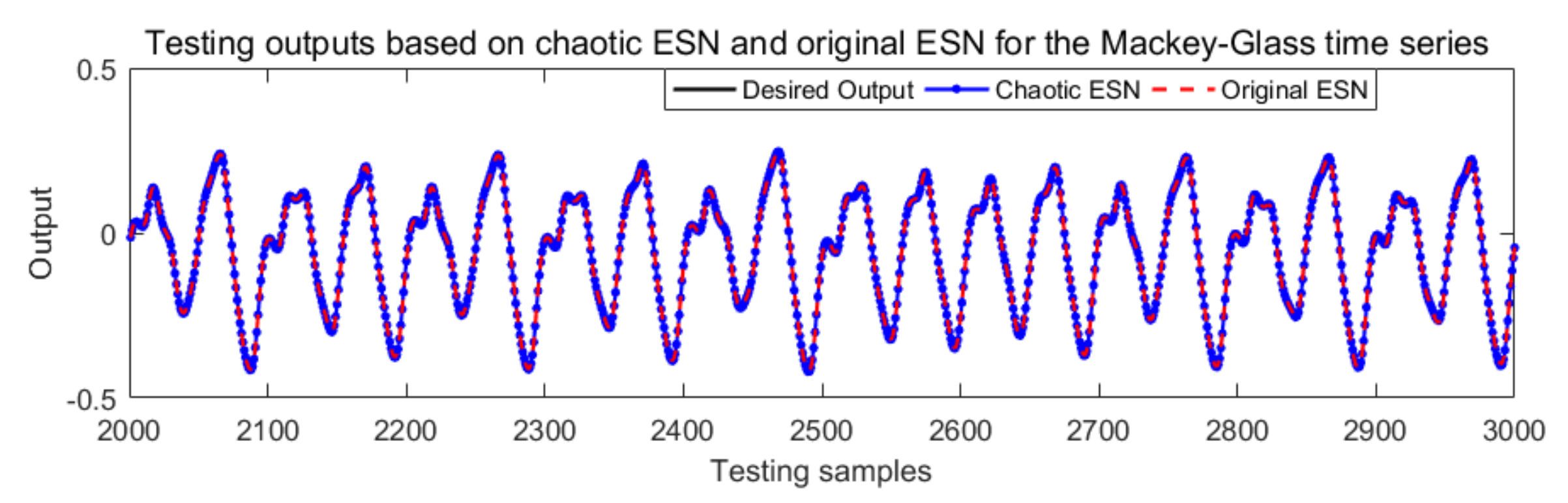}}
	\subfigure[Prediction error $O(t)-P(t)$ of the Mackey-Glass time series for the traditional ESN and chaotic ESN]{
	  \includegraphics[width=8cm]{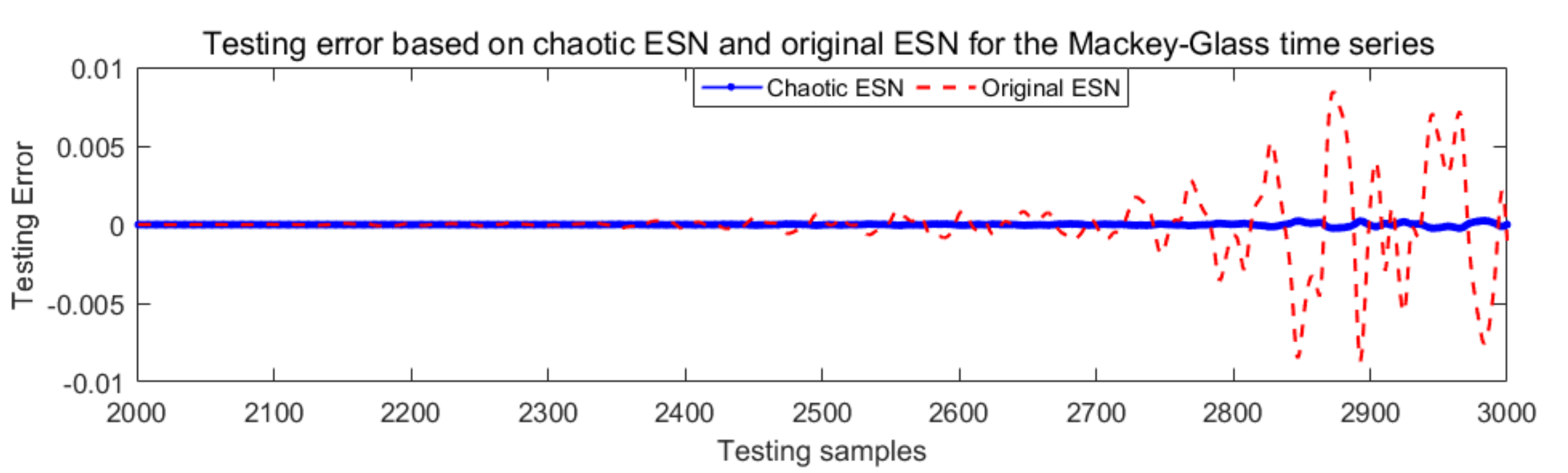}}
	\caption{Prediction and Prediction error $O(t)-P(t)$}
	\label{ESN for MGS prediction} 
  \end{figure}

We select the iterative functions $\langle 1 \rangle$, $\langle 2 \rangle$, $\langle 3 \rangle$, $\langle 11 \rangle$, $\langle 12 \rangle$, and $\langle 13 \rangle $ from Table~\ref{mddcs}, and get the corresponding HDDCSs as shown in the first column of Table~\ref{ddfc} ($s,u,v,w$ in a HDDCS are some independent random sequences). On the basis of the HDDCSs, the corresponding adjacency matrix $A_{M\times M}^{\langle 1 \rangle },A_{M\times M}^{\langle 2 \rangle },A_{M\times M}^{\langle 3 \rangle },A_{M\times M}^{\langle 11 \rangle },A_{M\times M}^{\langle 12 \rangle },A_{M\times M}^{\langle 13 \rangle }$ are obtained, where  $M={2^{mN}}$, $m$ is the dimension of ${G_F}$ and $N$ is the finite precision.  ${W_{S}}$ in Eq.~\eqref{eq4-3} is replaced with $A_{M\times M}^{\langle 1 \rangle },A_{M\times M}^{\langle 2 \rangle },A_{M\times M}^{\langle 3 \rangle },A_{M\times M}^{\langle 11 \rangle },A_{M\times M}^{\langle 12 \rangle },A_{M\times M}^{\langle 13 \rangle }$ respectively, so that the chaotic ESNs are constructed and then the root mean square error (RMSE) are obtained as shown in the Table~\ref{ddfc}, where None means this adjacency matrix ${{A}_{M\times M}}$ corresponding to the above HDDCS can not construct that size of the reservoir. Note that the order of the square matrix in Table~\ref{ddfc} is very large, and it is not convenient to give a specific $A_{M\times M}^{\langle i \rangle }$, so it is omitted here.

In Table~\ref{ddfc}, RMSE is

\begin{equation}
	\text{RMSE}=\sqrt{\frac{1}{1000}\sum\limits_{t=2001}^{3000}{{{(P(t)-O(t))}^2}}}  	
\end{equation}    
According to Table~\ref{ddfc}, as $ M = {2 ^ {mN}} $ increases, RMSE decreases, which indicates that the accuracy of the predictions is improved as  $M = {2 ^ {mN}} $ increases. On the base of given $ M = {2 ^ {mN}} $, higher-dimensional systems have higher prediction accuracy than lower-dimensional systems.

\section{Conclusion}
The literature~\cite{Wang:hddcs:IEEECSI2016} introduced and studied HDDCS and proved that HDDCS is chaotic in the Devaney's definition of chaos. This article further provides the  proof  of topological  mixing for HDDCS. In particular, this article constructs an iterative function $F$ according to the loop-state contraction algorithm, so that the state transition diagram of $G_F$ in HDDCS is strongly connected, thus a general design method for constructing HDDCS is found. The adjacency matrix corresponding to HDDCS constructed by the loop-state contraction algorithm are applied to build the chaotic ESN, and the Mackey-Glass time series is predicted, and relatively good prediction results are obtained. In short, this article expands the theoretical research of HDDCS, and also provides more construction methods and practical applications for HDDCS.


\begin{thebibliography}{99}

	\bibitem{Li:ITCSI2019}
	C. Li, B. Feng, S. Li, J. Kurths and G. Chen, ``Dynamic analysis of digital chaotic maps via state-mapping networks,'' \emph{IEEE Trans. Circuits Syst. I}, vol.~66, no.~6, pp. 2322--2335, 2019.

	\bibitem{Li:Cryptanalysis:IA2018}
	C. Li, D. Lin, B. Feng, J. L\"u and F. Hao, ``Cryptanalysis of a chaotic image encryption algorithm based on information entropy,'' \emph{IEEE Access}, pp. 75834--75842, 2018.
	
	\bibitem{Hua:Design:IS2017}
	Z. Hua and Y. Zhou, ``Design of image cipher using block-based scrambling and image filtering,'' \emph{Inform. Sciences}, vol.~396, pp. 97--113, 2017.
	
	\bibitem{Ye:An:N2017}
	G. Ye and X. Huang, ``An efficient symmetric image encryption algorithm based on an intertwining logistic map,'' \emph{Neurocomputing}, vol.~251, pp. 45--53, 2017.
	
	\bibitem{Lambic:A:ND2017}
	D. Lambi\'c, ``A novel method of S-box design based on discrete chaotic map,'' \emph{Nonlinear Dynam.}, vol.~87, no.~4, pp. 2407--2413, 2017.
	
	\bibitem{Xie:On the cryptanalysis:SP2017}
	E. Xie, C. Li, S. Yu and J. L\"u, ``On the cryptanalysis of Fridrich's chaotic image encryption scheme,'' \emph{Signal Process.}, vol.~132, pp. 150--154, 2017.
	
	 \bibitem{MAY:Logistic:Nature1976}
	R.~M. May, ``Simple mathematical models with very complicated dynamics, '' \emph{Nature}, vol. 261, no. 5560, pp. 459--467, 1976.
	
	\bibitem{Chen:Design and FPGA-based:ITCSI2017}
	S. Chen, S. Yu, J. L\"u, G. Chen and J. He, ``Design and FPGA-based realization of a chaotic secure video communication system,'' \emph{IEEE Trans. Circuits Syst. Video Technol.}, vol.~28, no.~9, pp. 2359--2371, 2018.

	\bibitem{Zhou:Hidden:IJBC2018}
	W. Zhou, G. Wang, Y. Shen, F. Yuan and S. Yu, ``Hidden coexisting attractors in a chaotic system without equilibrium point,'' \emph{Int. J. Bifurcation and Chaos}, vol.~28, no.~10, pp. 1830033, 2018.
	
	 \bibitem{LiShujun:Rules:IJBC2006}
	G.~Alvarez and S.~Li, ``Some basic cryptographic requirements for chaos-based
	cryptosystems,'' \emph{Int. J. Bifurcation and Chaos},
	vol.~16, no.~8, pp. 2129--2151, 2006.
	
	 \bibitem{Hua:Sine:ITIE2018}
	Z. Hua, B. Zhou and  Y. Zhou, ``Sine chaotification model for enhancing chaos and its hardware implementation,'' \emph{IEEE Trans. Ind. Electron.}, vol.~66, no.~2, pp. 1273--1284, 2019.

	\bibitem{Liu:An:ITCSI2016}
	H. Liu, H. Wan and KT. Chi and J. L\"u, ``An encryption scheme based on synchronization of two-layered complex dynamical networks,'' \emph{IEEE Trans. Circuits Syst. I}, vol.~63, no.~11, pp. 2010--2021, 2016.
	
	
	 \bibitem{Yu:Design Principles:2018}
	S. Yu, \emph{Design Principles and Applications of New Chaotic Circuits and Systems}. Beijing: Science Press, pp. 270--303, 2018. (in Chinese)
	
	 \bibitem{Lai:Dynamic:CSF2018}
	Q. Lai, B. Norouzi and F. Liu, ``Dynamic analysis, circuit realization, control design and image encryption application of an extended L\"u system with coexisting attractors,'' \emph{Chaos Soliton. Fract.}, vol.~114, pp.230--245, 2018.

	\bibitem{Luo:A:ND2018}
	Y. Luo, R. Zhou, J. Liu, Y. Cao and X. Ding, ``A parallel image encryption algorithm based on the piecewise linear chaotic map and hyper-chaotic map,'' \emph{Nonlinear Dynam.}, vol.~93, no.~3, pp. 1165--1181, 2018.
	
	\bibitem{Liu:Counteracting:IJBC2017}
	Y. Liu, Y. Luo, S. Song, L. Cao, J. Liu and J. Harkin, ``Counteracting dynamical degradation of digital chaotic Chebyshev map via perturbation,'' \emph{Int. J. Bifurcation and Chaos}, vol.~27, no.~3, pp. 1750033, 2017.
	
	\bibitem{Luo:encryption:IA2019}
	Y. Luo, X. Ouyang, J. Liu and L. Cao, ``An image encryption method based on elliptic curve elgamal encryption and chaotic systems,'' \emph{IEEE Access}, vol.~7, pp. 38507--38522, 2019.
	
	\bibitem{Wang:On fuzzy:ITC2014}
	Z. Wang and H. Wu, ``On fuzzy sampled-data control of chaotic systems via a time-dependent Lyapunov functional approach,'' \emph{IEEE Trans. Cybernetics}, vol.~45, no.~4, pp. 819--829, 2015.
	
	\bibitem{Guyeux:Hash:JACT10}
	C.~Guyeux and J.~M. Bahi, ``Hash functions using chaotic iterations, '' \emph{J. Algorithm. Computat. Technology}, vol.~4, no.~2, pp.
	167--182, 2010.
	
	\bibitem{Bahi:XORshift:JNCA04}
	J.~M. Bahi, X.~Fang, C.~Guyeux, and Q.~Wang, ``Suitability of chaotic
	iterations schemes using xorshift for security applications, '' \emph{J. Netw. Comput. Appl.}, vol.~37, pp. 282--292, 2014.

	\bibitem{Bahi:PRNS:IJAS11}
	------, ``Evaluating quality of chaotic pseudo-random generators. application
	to information hiding, '' \emph{Int. J. Advances in Security}, vol.~4, no.
	1-2, pp. 118--130, 2011.
	

	
	\bibitem{SMYu:integer:IJBC2014}
	Q.~Wang, S.~Yu, C.~Guyeux, J.~M. Bahi, and X.~Fang, ``Theoretical design and
	circuit implementation of integer domain chaotic systems, '' \emph{Int. J. Bifurcation and Chaos}, vol.~24, no.~10, pp. 1450128, 2014.

	\bibitem{Wang:DCS:book18}
	Q.~Wang, S.~Yu, C.~Guyeux, \emph{Design of Digital Chaotic Systems Updated by Random Iterations}. Cham: Springer, pp. 35--45, 2018.
	
	\bibitem{Devaney:Chaos:2003}
	R.~L. Devaney, \emph{An Introduction to Chaotic Dynamical Systems}. Colorado: Westview Press, pp. 48--53, 2003.
	
	\bibitem{Wang:hddcs:IEEECSI2016}
	Q.~Wang, S.~Yu, C.~Li, J.~L\"u, X.~Fang, C.~Guyeux, and J.~M. Bahi ``Theoretical design and FPGA-based implementation of higher-dimensional digital chaotic systems, '' \emph{IEEE Trans. Circuits Syst. I}, vol.~23, no.~3, pp. 401--412, 2016.
	
	\bibitem{Xiong:Chaos:ICTP1991}
	J. Xiong and Z. Yang, ``Chaos caused by a topologically mixing map,'' \emph{World Sci. Adv. Ser. Dyn. Syst.}, vol.~9, pp. 550--572, 1990.
    
    \bibitem{Zhou:Dynamics:1997}
	Z. Zhou, \emph{Symbolic Dynamics, Advanced Series in Nonlinear Science}. Shanghai: Shanghai Scientific and Technological Education Publishing House, 1997. (in Chinese)
	
	

	\bibitem{Dharwadker:graphtheory2007}
	A. Dharwadker and S. Pirzada, \emph{Applications of graph theory}. California: Institute of Mathematics, pp. 300--310, 2007.
	
	\bibitem{Khuller:sccnon1998}
	S. Khuller, \emph{Approximation algorithms for finding highly connected subgraphs (chapter 6)}. Boston: PWS, pp. 1--32, 1996.


	\bibitem{Jaeger:ESN:2001}
	H. Jaeger, ``The ``echo state'' approach to analysing and training recurrent neural networks-with an erratum note,'' \emph{GMD Tech. Rep.}, vol.~148, no.~34, pp. 13, 2001. 

	\bibitem{Haridas:ESN:2018}
	A.~V. Haridas, R. Marimuthu and V.~G. Sivakumar, ``A critical review and analysis on techniques of speech recognition: The road ahead,'' \emph{Int. J. Knowl.-Based Intell. Eng. Syst.}, vol.~22, no.~1, pp. 39--57, 2018. 

	\bibitem{Morse:ESN:2017}
	A.~F. Morse and A. Cangelosi, ``Why are there developmental stages in language learning? A developmental robotics model of language development,'' \emph{Cognitive Sci.}, vol.~41,  pp. 32--51, 2017. 

	\bibitem{Yao:ESN:2018}
	X. Yao, Z. Wang and H. Zhang, ``Identification method for a class of periodic discrete-time dynamic nonlinear systems based on Sinusoidal ESN. Neurocomputing,'' \emph{Neurocomputing}, vol.~275, pp. 1511--1521, 2018. 

	\bibitem{Jaeger:ESN:2004}
	H. Jaeger and H. Harald, ``Harnessing nonlinearity: Predicting chaotic systems and saving energy in wireless communication,'' \emph{Science}, vol.~304, no.~5667, pp. 78--80, 2004. 


	 \bibitem{Li:Chaotic:ITNNLS2012}
	D. Li, M. Han and J. Wang, ``Chaotic time series prediction based on a novel robust echo state network,'' \emph{IEEE Trans. Neural Netw. Learn Syst.},  vol.~23, no.~5, pp. 787--799, 2012.

	 \bibitem{Ren:Performance:ITC2020}
	H. Ren, H. Yin, C. Bai and J. Yao, ``Performance improvement of chaotic baseband wireless communication using echo state network,'' \emph{IEEE Trans. Commun.}, vol.~68, no.~10, pp. 6525--6536, 2020.


	\bibitem{Shi:MGS:Ieee2007}
	Z. Shi and M. Han, ``Support vector echo-state machine for chaotic time-series prediction,'' \emph{IEEE Trans. Neural Network}, vol.~18, no.~2, pp. 359--372, 2007.

	\bibitem{Bahi:Neural:2012}
    J.~M. Bahi, J. Couchot, C.~Guyeux and M. Salomon, ``Neural networks and chaos: Construction, evaluation of chaotic networks, and prediction of chaos with multilayer feedforward networks, '' \emph{AIP Chaos}, vol.~22, no.~1, pp. 013122, 2012.
    
    \bibitem{Li:Period:2004}
	T. Li and J. Yorke, ``Period three implies chaos,'' \emph{Am. Math. Mon.}, vol.~82, no.~10, pp. 985--992, 1975.

	\bibitem{Hsu:CellMap:IJBC92}
	C.-S. Hsu, ``Global analysis by cell mapping, '' \emph{Int. J. Bifurcation
	and Chaos}, vol.~2, no.~4, pp. 727--771, 1992.
	
	\bibitem{Shreim:NetworkCA:2007}
	A.~Shreim, P.~Grassberger, W.~Nadler, B.~Samuelsson, J.~E.~S. Socolar, and
	 M.~Paczuski, ``Network analysis of the state space of discrete dynamical
	 systems, '' \emph{Phys. Rev. Lett.}, vol.~98, no.~19, pp.  198701, 2007.

	\bibitem{Guyeux:Discrete:book13}
    C.~Guyeux and J.~M. Bahi, \emph{Discrete dynamical systems and chaotic machines: theory and applications}. CRC Press, pp. 53--57, 2013.	
	



\end{thebibliography}
\end{document}